\documentclass[sn-vancouver]{sn-jnl} 

\usepackage{graphicx}%
\usepackage{multirow}%
\usepackage{amsmath,amssymb,amsfonts}%
\usepackage{amsthm}%
\usepackage{mathrsfs}%
\usepackage[title]{appendix}%
\usepackage{xcolor}%
\usepackage{textcomp}%
\usepackage{manyfoot}%
\usepackage{booktabs}%
\usepackage{algorithm}%
\usepackage{algorithmicx}%
\usepackage{algpseudocode}%
\usepackage{listings}%
\usepackage{multirow}
\usepackage{booktabs} 
\usepackage{relsize}
\usepackage{bbm}
\usepackage{adjustbox}
\usepackage{graphicx}  
\usepackage{lmodern}
\usepackage{multicol}
\usepackage{fancyhdr}
\usepackage{placeins}  
\usepackage{amsthm}
\usepackage[scaled]{helvet}

\newtheorem{proposition}{Proposition}%
\newtheorem{lemma}{Lemma} 
\newtheorem{remark}{Remark}%

\theoremstyle{thmstylethree}%

\newcommand{\myrb}[1]{#1}

\begin{document}
	
	\title[The Life Care Annuity: enhancing product features and enriching perspectives]{The Life Care Annuity: enhancing product features and refining pricing methods} 
	
	
	\author*[1]{\fnm{Giovanna} \sur{Apicella}}\email{giovanna.apicella@uniud.it}
		\equalcont{These authors contributed equally to this work.}
	
	\author[1]{\fnm{Marcellino} \sur{Gaudenzi}}\email{marcellino.gaudenzi@uniud.it}
		\equalcont{These authors contributed equally to this work.}
		
		\author[1]{\fnm{Andrea} \sur{Molent}}\email{andrea.molent@uniud.it}
		\equalcont{These authors contributed equally to this work.}
	
	\affil*[1]{\orgdiv{Department of Economics and Statistics}, \orgname{University of Udine}, \orgaddress{\street{via Tomadini 30/A}, \city{Udine}, \postcode{33100},  \country{Italy}}}

	
	\abstract{The state-of-the-art proposes Life Care Annuities, that have been recently designed as variable
		annuity contracts with Long-Term Care payouts and Guaranteed Lifelong Withdrawal Benefits. In this paper, we propose more general features for these insurance products and refine their pricing methods. We name our proposed product ``GLWB-LTC''. In particular, as to the product features, we allow dynamic withdrawal strategies, including the surrender option. Furthermore, we consider stochastic interest rates, described by a Cox-Ingersoll-Ross process. As to the numerical methods, we solve the stochastic control problem involved by the selection of the optimal withdrawal strategy through a robust tree method, which outperforms the Monte Carlo approach. We name this method ``Tree-LTC'', and we use it to estimate the fair price of the product, as some relevant parameters vary, such as, for instance, the entry age of the policyholder. Furthermore, our numerical results show how the optimal withdrawal strategy varies over time with the health status of the policyholder. Our findings stress the important advantage of flexible withdrawal strategies in relation to insurance policies offering protection from health risks. Indeed, the policyholder is given more choice about how much to save for protection from the possible disability states at future times.
	}

	
	
	
	\maketitle
	\section{Introduction}\label{sec1}
	As stressed by the World Health Organization\footnote{https://www.who.int/news-room/fact-sheets/detail/ageing-and-health}, every country in the world is experiencing the phenomenon of population ageing, whose drivers are lower fertility and higher survival prospects. Indeed, older people represent a growing share of the population. For instance, it is estimated that the proportion of the world's population over 60 years will nearly double from 12\% to 22\% by 2050. 

While a longer life gives rise to several opportunities for older people and their families, the quality of the added years of life highly depends on health. 

Currently, it is estimated that more than 250 million people aged 60 years and over have moderate to severe disability\footnote{https://www.un.org/development/desa/disabilities/disability-and-ageing.html}. The progressive population ageing may lead to more people experiencing age-related diseases/disorders and disability in their more advanced stages of life \citep{Petretto2022}. Furthermore, as emphazised in \cite{OECD2022}, Long-COVID or ``Post COVID-19 Condition (PCC)'' will likely make chronic diseases more prevalent in both younger and older people in the coming years.

Ageing and disability cannot be disentangled. Enabling older people to receive care and support in the face of declines in physical and mental capacity (e.g., granting access to long-term-care) is indeed one of the targeted initiatives related to healthy ageing\footnote{https://www.who.int/initiatives/decade-of-healthy-ageing}, being aligned with the 17 Sustainable Development Goals (SDGs) set in the 2030 Agenda for Sustainable Development\footnote{https://sdgs.un.org/goals}. 

As populations grow old and the demand for LTC services is expected to increase in the coming years and decades, governments seek to balance financial sustainability with the provision of effective social protection against the financial hardship that may be caused to individuals by LTC needs  \citep{Costa2020}. Indeed, institutionalized care may be very expensive on a daily basis and may be needed on extended time horizons, thus implying dramatic costs. Public social protection systems play a fundamental role in subsidising the total costs of LTC in a large majority of OECD and EU countries, even for the people with higher incomes. 

Private insurance can complement or supplement the public sector, e.g., by extending care options and filling the gaps in public coverage \citep{OECD2021}. Tipically, long-term care insurance policies are designed to support the payment for
assistance (at home or in an institution) for individuals
who experience difficulty accomplishing ``activities of daily living'' (ADLS) because of physical and/or cognitive impairments. 
Although the potential need for long-term care represents one of the greatest financial risks for most older people and their families, private long-term care insurance has a relatively small market penetration in OECD countries \citep{OECD2011} and worlwide, with significant welfare and public policy implications. 

The literature has identified several reasons why individuals may decide not to purchase private LTC insurance coverage, addressing, for instance, the main demand-side factors that may drive such a behaviour \citep{Eling2019}. Among these (e.g., high premiums loadings, information asymmetry), as argued by \cite{Brown2009}, the individuals' lack of a proper understanding of LTC insurance and of the LTC expenditure risk contribute to impose limits on the size of the private LTC market. 

Health shocks are very difficult to predict, in terms of both their severity and the time when they occur. Liquidity needs due to perceived health cost risks have economic effects. For instance, health cost risk is offered as a possible explanation to the low annuitization rate being consistently observed in the private insurance market, 
the so-called ``annuitization puzzle'' \citep{Werker2017}. Nevertheless, \cite{Xu2023} show that, when health shocks are considered, access to LTC insurance mitigates the reduction in the annuity demand induced by a higher level of risk aversion. 

As discussed in \cite{Gatzert2023}, the issue of the optimal decumulation of wealth during retirement is highly relevant. The most recent paradigms analyse products and strategies for the decumulation of wealth under the perspectives of both insurers and retirees to ensure that demand meets supply, thus accounting also for risk perception and behavioral aspects.

The private insurance sector has explored the combination of LTC with other insurance products, e.g., annuities, so that to bundle LTC with other risks as, e.g., in \cite{Webb2009}. \cite{Getzen1988} proposed  ``longlife insurance'' plans, combining deferred annuity benefits, health insurance and LTC. Such insurance coverage was designed to match protection from the risks of chronic illness with protection from the risks of higher longevity and thus to mitigate the adverse selection affecting both LTC insurance and annuities. More recently, \cite{Fuino2022} evaluated a life annuity product with an embedded care option potentially supporting the financial needs of dependent persons, by accounting for both the insurer's perspective and the policyholder's willingness to pay for the care option.  

\cite{Murtaugh2001},  investigated the empirical features of ``life care annuity'', namely the combination of life annuity with LTC disability coverage at retirement. This product has the potential to extend disability protection to a wider segment of the population and to mitigate adverse selection, thus reducing its purchase cost. \cite{Brown2013} provided an empirical examination of life care annuity based on the data from the Health and Retirement Study (HRS).

One of the most recent innovations discussed in the state-of-the-art is the Variable Life Care Annuity with Guaranteed Lifetime Withdrawal Benefits (LCA-GLWB), protecting downside risk, through guaranteed income streams, together with longevity and LTC cost risk \citep{Hsieh2018}. Specifically, under the general scheme of a GLWB variable annuity contract, the policyholder makes a single lump sum payment, that is invested in risky assets, such as a mutual fund. The amount of the lump sum payment typically represents the benefit base, or guarantee account balance. The policyholder is allowed to withdraw a given fraction of the benefit base each year until she remains alive. The GLWB thus combines longevity protection, exposure to equity markets and withdrawal flexibility. The valuation of these guarantees and the involved technical problems are discussed, for instance, in \cite{Bacinello2011,Steinorth2015,Molent2016,DeAngelis2022}. Compared to a traditional GLWB variable annuity contract, a variable LCA-GLWB contract provides also LTC payouts if the policyholder incurs in defined frailty state levels (e.g., impairments in ADLS). The evaluation of such insurance contract requires tackling three sources of uncertainty: the occurrence of ADLs impairments, prospective longevity (either in the healthy or disabled condition), and the performance of the financial market. 

In \cite{Hsieh2018}, withdrawals are possible only at contractually defined percentages of the benefit base. Furthermore, pricing relies on Monte Carlo valuation methods, such as the variance reduction techniques (specifically, control variates technique). In such a pricing framework, the interest rate is not stochastic. 

The key idea of our paper is to provide more general features for the variable annuity contract with LTC payouts and GLWB and to refine its pricing methods. We denote our proposed product ``GLWB-LTC''. The three characteristics that make GLWB-LTC depart from the LCA-GLWB product of \cite{Hsieh2018}, relative to the product specification and pricing method: (i) stochastic interest rate model, namely Cox-Ingersoll-Ross (CIR) as in \cite{Cox1985}, (ii) dynamic withdrawal strategy, as in \cite{Forsyth2014}, (iii) pricing based on a tree method, as in \cite{Appolloni2015}. 

A stochastic framework for the interest rate model allows a more accurate description of the future evolution of interest rates, over the long time horizon implied by the policy duration. Specifically, the underlying fund is supposed to evolve, under a risk neutral measure, as a geometric Brownian motion, as in the Black-Scholes (BS) model, but with stochastic drift given by the short interest rate. This latter, is supposed to follow a CIR process, thus we term this the BS-CIR model. 

A dynamic withdrawal strategy allows the policyholder to choose the amount to be withdrawn. Accordingly, the benefit base may be increased if the policyholder withdraws no funds in a given year (i.e., bonus or roll-up). Furthermore, the contract may terminate if the policyholder opts for complete surrender, namely she withdraws the whole residual amount in the investment account. As illustrated in \cite{Bacinello2009}, insurance products embedding a surrender option may be more attractive to the demand side, as policyholders may be less prone to perceive insurance securities as illiquid investments. This early exercise feature acquires even more relevance in light of the fact that mis-perceptions of health cost and mortality-related risks may further contribute to make long-term contracts such as annuities and LTC schemes poorly attractive for individuals in their pre-retirement ages; see, e.g., \cite{ODea2023}. 

From a numerical point of view, the presence of a surrender option implies tackling an American-style  option enabling the policyholder to exit the contract and be paid the surrender value. We solve the stochastic control problem involved by the evaluation of this option, through an improved version of the tree method pricing technique in \cite{Appolloni2015},  as it is proven to be fast and efficient for pricing American options in the BS-CIR model, without any numerical restriction on
its parameters. The employed method, which we term \textit{Tree-LTC}, can be applied also in the case of high volatility of interest rates and shows advantages over Monte Carlo methods.

Finally, we perform several numerical experiments. As a first step, we validate the Tree-LTC method, showing that it outperforms the traditional Monte Carlo simulation approach when pricing a traditional LCA-GLWB insurance product. As a next step, we focus on our proposed GLWB-LTC insurance product and carry out its evaluation. Our numerical results describes how the fair prices and the optimal withdrawal strategy vary with some features of the policyholder, such as her age and her health status, and other factors such as market conditions (as expressed by the volatility of the fund and of  the interest rate). Our novel evidence shows an important advantage of flexible withdrawal strategies, in relation to insurance policies offering protection from health risks.  Indeed, against a small increase in the fee, the policyholder is given more choice about how much to save for protection from the possible disability states at future times.

The paper is structured as follows. Section 2 introduces the product and the model specifications. Section 3 presents in detail the principles and methods adopted for the evaluation of the contract under examination. Section 4 discusses the numerical results. Finally, Section 5 draws the conclusions.

	\section{Product and model specifications}\label{sec:methods}	
	In this Section, we illustrate how the GLWB-LTC product specification is designed and we describe the underlying modelling framework.

	\subsection{Health state model}
	LCA policyholders are characterized by complex mortality patterns. Health and mortality risks play a substantial role within the actuarial modeling of health and life insurance policies, and require a proper assessment, according to the regulatory framework of the Solvency II Directive\footnote{\url{https://www.eiopa.europa.eu/browse/regulation-and-policy/solvency-ii_en}} \citep{Shao2015}.

\cite{Pitacco1995} illustrates how, in the framework of the mathematics of Markov and semi-Markov stochastic processes, it is possible to develop a general approach for the 
actuarial modelling of disability and related benefits, such as LTC annuities. Indeed, the evaluation of life insurance policies with long term benefits is usually based on probabilistic structures consistent with Markovian multi-state models, such as, for instance, in \cite{Haberman1998,Menzietti2012, Tabakova2021}. Interdisciplinary literature proposes a variety of statistical methods to estimate transition matrices of Markov chains from data, for instance \cite{Baione2014,Helms2015}. 

In our paper, we use the disability model proposed by the authors in \cite{Manton1993, Pritchard2006}. According to their model, disability is defined in terms of loss of   \textit{instrumental activities
	of daily living} (IADL, such as meal preparation, grocery shopping, getting
around outside, using the telephone), and loss of \textit{activities of daily living} (ADL, such as eating, getting in and out of bed, getting around
inside, dressing, bathing, getting to the bathroom
or using the toilet). In particular, such a model includes seven health states: healthy, impairment in only IADL, 1-2 impairments in ADLs, 3-4 impairments in ADLs, 5-6 impairments in ADLs, institutionalized and dead. 

Here are the main features of this model. Let $M_t\in\left\lbrace1,2,3,4,5,6,7 \right\rbrace $  be a random variable which represents the health state of the policyholder (hereinafter PH) at time $t$, being $x_0$ her age at inception. Now, for $0\leq s\leq t$, we term $P^{x_0}(s,t)$ the $7\times7$ transition probability matrix with entries $$ p^{x_0}_{i,j}(s,t) =\mathbb{P}\left( M_{t}=j|M_{s}=i\right). $$ Transition rates can be used to define the process: let $Q^{x_0}(t)$ be the $7\times 7$ matrix, given by 
\begin{equation*}
	\begin{aligned}
		& q^{x_0}_{i, j}=\lim _{\Delta t \rightarrow 0} \frac{p^{x_0}_{i, j}(t, t+\Delta t)}{ \Delta t}, i \neq j, \\
		& q^{x_0}_{i, i}=-\sum_{j\neq i} q^{x_0}_{i, j}, i=1, \ldots, 7 .
	\end{aligned}
\end{equation*}
The matrices $Q^{x_0}$ are assumed to be time-homogeneous during each year, that is, for each $n\in \mathbb{N}$, $Q^{x_0}(t)=Q^{x_0}(s)$ holds for all $t,s$ such that $n\leq s\leq t< n+1$. 
Then, the transition probability matrix $P^{x_0}(n, n+1)$ between two anniversaries can be computed from the transition intensities via the matrix exponential operation, that is
$$
P^{x_0}(n, n+1)=e^{Q^{x_0}(n)}.
$$
Furthermore, we obtain transition intensities $Q^{x_0}(n)$ based on the parameter values shown in \cite{Pritchard2006} (Table 8, page 68), that were obtained by applying the penalized likelihood methodology to the interval-censored longitudinal data from the National Long-Term Care Study. For the sake of completeness, we report this Table in the Appendix \ref{appendix:ti}.  We stress out that this approach for modeling the health state of the PH is also adopted by \cite{Hsieh2018}, who, to the best of our knowledge, developed the most recent study on the evaluation of the LCA-GLWB insurance product. Using the same underlying transition matrices as in \cite{Hsieh2018} allows us to have a benchmark for validating some of the numerical outcomes shown in our paper and to propose original developments based on alternative product specifications and computation methods.

\begin{remark}
	The model by \cite{Pritchard2006} allows the generation of transition probabilities between health states for any age of the PH, without placing an upper limit on the age of the insured. Following common practice, see e.g. \cite{Forsyth2014} \cite{Molent2016}, we limit the maximum age of the insured to $122$. Consequently, whatever her health state at $121$, the probability of transition to health state $7$ is equal to $1$.
\end{remark}
	
	\subsection{Dynamics of the mutual fund and of the interest rate}
	Let us consider a risk-neutral measure $\mathbb{Q}$. The risk neutral dynamics of the stochastic processes describing the mutual fund $F_t$ and the interest rate $r_t$ are as follows:
	\begin{equation} \label{model}
	\begin{cases}
		dF_t=r_t F_t dt + \sigma_F F_t dW^1_t,\\
		dr_t=k_r(\theta-r_t) dt + \sigma_r \sqrt{r_t} dW^2_t,
	\end{cases}
	\end{equation}
	
	\noindent
	where the constant parameters $k_r, \theta$ and $\sigma_r$ are  the  rate of mean reversion, the long run mean and the   volatility of the interest rate, respectively. Furthermore, $W^1$ and $W^2$ are Brownian motions such that their correlation equals $\rho$.  
	
	\begin{remark}		
\myrb{	It is well-established in the literature the structural soundness of the CIR model in preventing negative rates under normal conditions, with no negative interest rates. Nevertheless, we remark that   the CIR model can be adapted to accommodate negative rates if necessary, through simple modifications such as shifting the rate distribution downward, making it a versatile tool in stochastic interest rate modeling (see, e.g.,  \cite{russo2017calibrating} and   \cite{Orlando2019}).  }
	\end{remark}
	\subsection{The GLWB-LTC}
	The GLWB-LTC insurance product guarantees the PH the right to make guaranteed withdrawals and can provide the payment of an annual disability allowance.
	
	At time $t=0$, the PH purchases the product through an initial one-off payment, which we denote by $P$. This amount determines the initial value of the two indicators governing the evolution of the contract: the account value $A$ and the benefit base $B$. In particular, the account value is used to calculate the maximum withdrawable amount, as well as the death benefit. The benefit base, on the other hand, governs the payments guaranteed by the contract, such as the LTC benefits and the minimum amount withdrawable by the PH. The state parameters $A$ and $B$ are two stochastic processes defined for each time instant between $t=0$ and $t=\tau$, the first anniversary of the contract inception following the insured's death. 
	
	\subsubsection{Initiation of the contract}
	The initial values of $A$ and $B$, denoted by $A_0^-$ and $B_0^-$ respectively, are both set equal to $P$:
	\[A_0^-=B_0^-=P.\]
	
	\myrb{Immediately after the initiation of the contract, the account value is charged with some specified fees, while the benefit base remains unaffected. Adopting the approach proposed by \cite{Hsieh2018}, the decrement in the account's value, owing to these fees, is regulated by two parameters, \(\alpha\) and \(\beta\). These parameters are indicative of the annual costs per unit for \(A\) and \(B\) respectively. Consequently, on each anniversary of the contract,  it is reduced by \(\alpha A\)   and \(\beta B\) as long as the account value remains positive.
	} Specifically, if we denote by $A_0^{1+}$ and $B_0^{1+}$ the value of $A$ and $B$ immediately after the fees are taken, the following holds:
	\[A_{0}^{1+}=\max\left(A_{0}^{-}-\alpha A_{0}^{-}-\beta B_{0}^{-},0\right),\quad B_{0}^{1+}=P. \]

At each anniversary $n$, hereafter, we will denote by $A_{n}^{2+}$ and $B_{n}^{2+}$ the values of $A$ and $B$ after the payment of the LTC to the PH and, then, by $A_{n}^{3+}$ and $B_{n}^{3+}$ the values of $A$ and $B$ after a withdrawal contingent on the choice of the PH at time $n$.

At contract inception, i.e., $n=0$, no LTC is paid to the PH in case of disability. Moreover,	
	the PH is not entitled to make any withdrawal. Therefore, neither $A$ nor $B$ is altered by a payment to, or a withdrawal from, the PH. Accordingly, the following holds:
	\[A_0^{3+}=A_0^{2+}=A_0^{1+},\quad B_0^{3+}=B_0^{2+}=B_{0}^{1+}.\]

	\subsubsection{Evolution of the contract between two anniversaries}
	\myrb{During the time between the beginning of the contract and the first anniversary, and similarly between any two consecutive anniversaries, } the account value $A$ varies  in proportion to the underlying fund, while the benefit base $B$ does not change: for all $t\in \left]0, 1\right[ $ it holds
	\begin{equation} \label{eq:1}
		\frac{dA_t}{A_t}=\frac{dF_t}{F_t}, \quad dB_s=0.
	\end{equation}  
	\noindent
	This holds also between any other two consecutive anniversaries.
	
	\subsubsection{Anniversary events if the PH is alive}
	On the first anniversary, but more generally on a generic anniversary thereafter, certain clauses of the contract are activated according to the PH's health status. Let $t$ represent the time of the $n$-th anniversary and let $A_{n}^{-}$ and $B_{n}^{-}$ be the values of $A$ and $B$ immediately before such a time. Thus, according to \eqref{eq:1}, we have 
	\[A_{n}^{-}=A_{n-1}^{3+}\cdot \frac{F_n}{F_{n-1}},\quad B_n^-=B_{n-1}^{3+}.\]
	
	\noindent
At a generic anniversary $n>0$, the account value is deduced by the fees. \\ Two payments can be received by the PH: the LTC benefit and the amount arising from the PH's withdrawal. Specifically, the   LTC benefit and the guaranteed minimum amount for withdrawal  are computed proportionally to the inflation-indexed benefit base and reduce the account value. We formalize the dynamics of $A$ and $B$ as follows.

	\begin{enumerate}
		\item Fees reduce the account value and do not alter the benefit base:
	\[A_{n}^{1+}=\max\left(A_{n}^{-}-\alpha A_{n}^{-}-\beta B_{n}^{-},0\right),\quad B_{n}^{1+}=B_{n}^{-}. \]
	\item The PH receives the LTC payment, $L_n(M_n)$, if her health state at this time corresponds to a disability condition covered by the contract. The amount of the LTC protection is proportional to the benefit base and is indexed by an inflation rate, denoted by $\pi$, as follows:
	
	\begin{equation}\label{eq:Theta}
		 L_{n}\left(M_{n}\right)=
		 \begin{cases}
		 	0 & \mathrm{if}\ M_{n}\in\left\{ 1,2,3\right\},\\
		 	cB_{n}^{1+}\left(1+\pi\right)^{n} & \mathrm{if}\ M_{n}\in\left\{ 4,5,6\right\}.
		 \end{cases}
	\end{equation}  
Accordingly, the values of $A$ and $B$ after the LTC payment are given by:
	\[A_{n}^{2+}=\max\left(A_{n}^{1+}-L_t\left(M_{n}\right),0\right),\quad B_{n}^{2+}=B_{n}^{1+}. \]
	\item
	Guaranteed withdrawals $G_n$ from the account value are possible at contractually defined percentages $g$ of the inflation-indexed benefit base, as follows: 
	\begin{equation}\label{eq:G}
		G_{n}=g \left(1+\pi\right)^{n} B_{n}^{2+},
	\end{equation}
	but	the PH  may also withdraw more or less than $G_n$. In this regard, let
	$W_{n}\in\left[0,\max\left(A_{n}^{2+},G_{n}\right)\right]$ be the amount
	that the PH withdraws. We stress out that 
	the maximum admissible withdrawal is given by $\max\left(A_{n}^{2+},G_{n}\right)$, that is the greater between the account value after the payment of the LTC and the minimum guaranteed withdrawal. 
\end{enumerate}
	In order to distinguish whether or not $W_{n}$ exceeds the guaranteed amount, we make use of an auxiliary parameter $\gamma$, as in   \cite{Forsyth2014}, whose value expresses the choice made by the PH with respect to the amount to be withdrawn at anniversary $n$. Specifically, $W_{n}$ is controlled
	by the parameter $\gamma_{n}\in\left[0,2\right]$ as follows:
	\[
	W_{n}=\begin{cases}
		\gamma_{n}G_{n} & \text{if }\gamma_{n}\leq1,\\
		\left(2-\gamma_{n}\right)G_{n}+\left(\gamma_{n}-1\right)  A_{n}^{2+}  & \text{if }\text{\ensuremath{\gamma_{n}>1}.}
	\end{cases}
	\]
 While $W_{n}$ represents the chosen withdrawal by the PH, let us denote by $Y_{n}$ the actual amount the PH receives, at time $t$. In this respect, we distinguish three cases:
	
	\begin{itemize}
		\item if $\gamma_{n}=0$, no money is withdrawn from the account. In this case, the PH renounces making a withdrawal and she is rewarded with a proportional bonus $b$ that increases the benefit base. Specifically:
		\begin{align*}
			Y_{n}=W_{n} & =0,\\
			A_{n}^{3+} & =A_{n}^{2+},\\
			B_{n}^{3+} & =B_{n}^{2+}\left(1+b\right).
		\end{align*}
		\item if $0<\gamma_{n}\leq1$ the performed withdrawal is less than or equal to the
		minimum guaranteed one (the latter case corresponds to  $\gamma_{n}=1$):
		\begin{equation}\label{eq:gamma01}
				\begin{aligned}
			Y_{n}=W_{n} &=\gamma_{n} G_{n}^{2+},\\
			A_{n}^{3+} & =\max\left(A_{n}^{2+}-W_{n},0\right),\\
			B_{n}^{3+} & =B_{n}^{2+}.
		\end{aligned}
		\end{equation}

		\item if $1<\gamma_{n}\leq 2$ the performed withdrawal is greater than the minimum
		guaranteed one. A proportional cost $\kappa_{n}$ is applied to the part of the
		withdrawal exceeding the guaranteed amount, this implying that the amount $Y_{n}$ that is actually received by the PH is smaller than $W_{n}$: 
		\begin{equation}\label{eq:gamma12}
		\begin{aligned}
			W_{n} & =\left(2-\gamma_{n}\right)\cdot G_{n}+\left(\gamma_{n}-1\right)A_{n}^{2+},\\
			Y_{n} & =G_{n}+\left(W_{n}-G_{n}\right)\left(1-\kappa_{n}\right),\\
			A_{n}^{3+} & =\max\left(A_{n}^{2+}-W_{n},0\right),\\
			B_{n}^{3+} & = B_{n}^{2+}\left(2-\gamma_{n}\right).
		\end{aligned}
	\end{equation}
 Usually, the cost $\kappa_{n}$ decreases over time and goes down to zero after a few years. 
	
\hskip0.5cm
We stress out that the case $\gamma_{n}=2$ implies total lapse and the end of the contract.  In this particular case, 
	$$ W_{n}=A_{n}^{2+},\ Y_{n} =G_{n}+\left(A_{n}^{2+}-G_{n}\right)\left(1-\kappa_{n}\right),\ A_{n}^{3+}=B_{n}^{3+}=0.$$  
	We denote by $\ell$ the anniversary, if it exists, such that $\gamma_\ell=2$. If $\gamma_{n}$ is always different from $2$, we define $\ell=+\infty$.
	
	\end{itemize}
\subsubsection{Anniversary events if the PH is dead}

		If the PH has died during the last year, i.e. $n=\tau$, her heirs receive a death benefit, calculated as follows, and the contract ends:
	$$Y_{\tau}	=W_{\tau}=g_{\tau}B_{\tau}^{-}+\max\left(0,A_{\tau}^{-}-g_{\tau}B_{\tau}^{-}\right),\ 
	A_{\tau}^{+}	=B_{\tau}^{+}=0.$$
 
We stress out that the contract may be terminated for two reasons: total lapse, or the death of the PH.
If we denote with $T$ the anniversary of contract termination, then $T=\min\left(\tau, \ell \right)$.

 \begin{remark}
 	Fees are paid since time $n=0$; the first withdrawal takes
 	place at time $n=1$. No fees are paid at the first anniversary after the death time, and no LTC payments are made as well ($L_\tau(7)=0$).
 	
 \end{remark}

	\section{Pricing the GLWB-LTC contract} 
	The value of the contract at any time $t$ depends on four state variables, namely $A_t$, $B_t$, $r_t$ and $M_t$, so we denote it as a function of these four state variables by $\mathcal{V}(A_t,B_t,r_t,M_t,t)$. In addition, at the $n$-th anniversary,   we  write $n^{-}$, $n^{1+}$, $n^{2+}$ and $n^{3+}$ to indicate the value of the contract just before the $n$-th anniversary, after the withdrawal of fees, after the payment of the LTC and after the withdrawal of the annuity, respectively.

	\subsection{Withdrawal strategy}
	The withdrawal strategy performed by the PH is a crucial point in the evaluation of the contract. Following the classification introduced by \cite{Bacinello2011}, we consider three particularly relevant strategies: ``static", ``mixed" and ``dynamic withdrawal". \myrb{Moreover, we also investigate a fourth strategy, termed ``full dynamic".} 
	
	\medskip Under the static withdrawal strategy, the PH has only one choice, that resides in withdrawing the minimum guaranteed sum, i.e., $\gamma=1$, at each anniversary in which the PH is alive. This static strategy is the only one considered in \cite{Hsieh2018}. In this particular case, the  benefit base never changes and is always equal to the premium $P$ paid by the PH at time zero. Consequently, the fees associated with the benefit base are constant at each anniversary and equal to $\beta P$.
	
		According to risk neutral valuation, under the static withdrawal strategy, the initial value of the contract is the expected value of future cash flows:
	$$ \mathcal{V}(P,P,r_0,M_0,0^{-})=\mathbb{E}^{\mathbb{Q}}\left[ \sum_{n=1}^{\tau} e^{-\int_{0}^{n} r_s\,ds} \left(L_{n}(M_{n})+ Y_{n}\right) \right].  $$

	\medskip The mixed strategy implies that the PH continues to draw at the guaranteed minimum rate until she dies or decides to terminate the contract early. Compared to the
	static strategy, there is thus the possibility of a total lapse, which can be achieved by choosing $\gamma=2$. 
	
			In the case of the mixed withdrawal strategy, the initial value of the contract is the expected value of future cash flows, obtained by using the optimal stopping strategy:
	\begin{equation} \label{eq:2}
		\mathcal{V}(P,P,r_0,M_0,0^{-})=\max_{\ell\in \mathcal{T}} \mathbb{E}^{\mathbb{Q}}\left[ \sum_{n=1}^{\min(\tau,\ell)} e^{-\int_{0}^{n} r_{s} \,ds} \left(L_{n}(M_{n})+ Y_{n}\right)\right],
	\end{equation}  
	 with $\mathcal{T}$ the set of optimal stopping times.
	 The optimal stopping time $\ell$ can easily be computed by means of dynamic programming. Specifically, $\ell$ is the first anniversary such that the value of the whole position in case of total lapse is larger than the continuation value. Let us write $A_{n}^{3+}\left(\gamma_{n} \right), B_{n}^{3+}\left(\gamma_{n} \right), W_{n}^{3+}\left(\gamma_{n} \right)$ and $ Y_{n}^{3+}\left(\gamma_{n} \right)$ to denote the values of $A_{n}^{3+}, B_{n}^{3+}, W_{n}^{3+}$ and $Y_{n}$ for a specific value of $\gamma_{n}$.  Then:
	 $$ \ell=\min\left\lbrace n=1,\dots,\tau-1\ \mathrm{s.t.}\ Y_{n}(2)\geq Y_{n}(1)+\ 
	 \mathcal{V}\left( A_{n}^{3+}(1),B_{n}^{3+}(1),r_{n},M_{n},n^{3+}\right) \right\rbrace. $$ 
	\noindent
	
	\medskip 
	
	Under the dynamic withdrawal strategy, the PH can freely choose the value of $\gamma$ $\in [0,2]$, for each withdrawal opportunity. Then, she can choose not to withdraw, or to withdraw more or less than the minimum guaranteed amount, with the maximum withdrawal implying the early termination of the contract. Equation \eqref{eq:2} also holds in this case. Here, we suppose that the PH chooses the value of $\gamma_{n}$
	that maximizes the total wealth she received, so that the value of $\gamma_{n}$ is defined as:
	$$ \gamma_{n} = \mathrm{arg}\max_{\gamma\in[0,2]} \left[ Y_{n}(\gamma)+\mathcal{V}\left( A_{n}^{3+}(\gamma),B_{n}^{3+}(\gamma),r_{t},M_{n},n^{3+}\right)\right] .$$
	
	\myrb{Finally, let us consider the full dynamic strategy, which extends the dynamic strategy by admitting that PH can perform total surrenders even in the time between two anniversaries. In this case, as usual, at any time $t\geq 0$  which is not an anniversary, 	the following equation holds
		\[\mathcal{V}\left( A_{t} ,B_{t},r_{t},M_{t},t\right)=\max\left\lbrace A_t\left(1-\kappa_{\left\lfloor t\right\rfloor } \right), \mathcal{C}\left( A_{t} ,B_{t},r_{t},M_{t},t\right)   \right\rbrace, \]
		where $\mathcal{C}\left( A_{t} ,B_{t},r_{t},M_{t},t\right) $ is the continuation value, that is the expected value of future discounted cash flows if the surrender option is not exercised at time $t$.
	}
	\begin{remark}
		The criterion for selecting the optimal strategy is based on maximizing the expected value, under risk-neutral probability, of the payment from the insurer to the insured. Alternatives have been proposed in the literature (see e.g. \cite{Choi2017} or \cite{Moenig2021}), which are based, for example, on maximizing expected utility. Our model can be adapted to consider these cases as well. In particular, in this case, the value of the policy should be calculated separately according to the insurer and according to PH. The latter determines the optimal withdrawal strategy, which is then used in the assessment of the cost of cover according to the insurer.
	\end{remark} 
	
		\begin{remark}
		
		\cite{Bacinello2022} prove by backward induction  that, if the optimality criterion is the maximisation of the value of total wealth, the optimal exercise strategy always consists of one of these three actions: to withdraw nothing, to withdraw the guaranteed minimum amount or to withdraw the maximum possible (by ending the contract). Such a feature of GLWB contracts is also known as the \textit{bang bang condition}. From the results of the numerical experiments, we found the same result for our product. Although the numerical method we propose has no difficulty in handling even intermediate withdrawal values, limiting the choice of possible range values would lead to a more efficient numerical procedure.
	\end{remark}

\begin{remark}
	We calculate the price of our insurance products based on a risk-neutral valuation approach. In this framework, we assume that the risks associated with death and disability can be diversified, as supported by \cite{ Milevsky2006}. If this assumption does not hold, the risk-neutral valuation can be modified through an actuarial premium principle, as noted by \cite{Gaillardetz2011}. 
\end{remark}

	\subsection{Similarity reduction}\label{sr}
	GLWB-type variable annuities are interesting from a computational point of view as their value is proportional to the ratio of the account value to the benefit base.   In mathematical terms, for every positive constant $\eta$, 
	\[\eta \cdot \mathcal{V}\left(A_{t}, B_{t}, r_{t}, M_{t},t \right)=\mathcal{V}\left(\eta   A_{t}, \eta   B_{t}, r_{t}, M_{t},t \right). \]
	This property, which has already been exploited in the literature (see e.g. \cite{shah2008}, \cite{Forsyth2014} or \cite{Molent2016}) also applies to the contract we consider in this paper, since all cash flows are proportional to the account value and to the  benefit base.	
	This useful property makes it possible to reduce the size of the problem, assuming $B$ to be constantly equal to its initial value $P$. In fact, taken $\eta=\frac{P}{B_{t}}$, one obtains
		\[  \mathcal{V}\left(A_{t}, B_{t}, r_{t}, M_{t},t \right)=\frac{B_{t}}{P}\cdot\mathcal{V}\left(\frac{A_{t}}{B_{t}}   P, P, r_{t}, M_{t},t \right). \]
		As a result, the evaluation of the contract is more efficient from a numerical point of view.
	
	\subsection{The numerical method}
 The numerical method, termed Tree-LTC, we propose to evaluate the GLWB-LTC contract in the BS-CIR model is an adapted and improved version of the tree model introduced by \cite{Appolloni2015}, that is suitable to our purposes in that it allows the evaluation of American derivative instruments in the considered stochastic model framework. Furthermore, the method in \cite{Appolloni2015} proves to be robust and stable from a numerical point of view. In a nutshell, the method constructs two trees that discretize the short interest rate and the underlying respectively. Subsequently, these structures are combined to obtain a two-dimensional tree. The transition probabilities relative to the nodes of the tree are computed by matching the conditional mean and the conditional covariance between the continuous and the discrete processes.

 \subsubsection{The tree for the interest rate}\label{ss:tree_for_r}
 The first step of the algorithm is to create a lattice to discretize the stochastic rate $r$.  \cite{Appolloni2015} suggest using a variation of the tree proposed by \cite{Nelson1990}, which  matches a first-order approximation of the first two moments of the process $r$. Such a tree works rather well when the maturities involved are relatively short, but \myrb{the computational cost can become high in the case of long maturities, such as those involved by our product. So here we propose an updated version of that tree that allows us to limit the number of nodes considered in the discretization. In practice, thanks to the properties of the CIR process,  it is necessary to consider only nodes between zero and a maximal value that depends only on the discretization step, in order to obtain a Markov chain that converges weakly to the continuous process $r$.}

 Specifically, \myrb{we consider a binomial tree which is used to define a  Markov chain which  matches a suitable approximation of  the first and the second moment of the continuous time process $r$. This feature guarantees weak convergence to the CIR process, as reported by \cite{Nelson1990}.} First of all, let $T\in \mathbb{N}$ be the maximum duration in years of the GLWB-LTC contract. For example, if the age of the PH at contract inception is $60$, then $T=122-60=62$.  Let us divide such period in $N T$ time steps, so that the time increment is $\Delta t=1/N$. We approximate the process $r$ in $[0,T]$  with a discrete time process $\bar{r}=\left\lbrace \bar{r}_{i}\right\rbrace_{i=0,\dots,NT}$, so that $\bar{r}_{i}$ approximates $r_{i\Delta t}$. The possible values  of the process $\bar{r}$ are defined as follows: \myrb{for $i=0,1,\dots,NT$ and  $k=0,1,\dots,i$   we set 
 with
 $$ \mathcal{R}_{i,k}=\left( \max \left( \sqrt{r_0} + (2k-i)\sigma_r\sqrt{\Delta t} ,0 \right) \right) ^2.
 $$
 In particular, we observe that, if $i$ is even,  the initial interest rate $r_0$ is included among these values as $\mathcal{R}_{i,i/2}=r_0$. Moreover, we observe that, if we set 
 $$ \underline{k}(i) =\left\lfloor  \frac{i}{2}-\frac{1}{\sigma_r}\sqrt{\frac{r_0}{\Delta t}}  \right\rfloor,$$ 
 then for all values $k=0,\dots,\underline{k}(i)$, it holds that $\mathcal{R}_{i,k}=0$, so one can consider only the values $k= k_{\mathrm{min}}(i) ,\dots, i$, where $k_{\mathrm{min}}(i)=\max\left\lbrace 0,\underline{k}(i)\right\rbrace$. With respect to \cite{Appolloni2015},  we thus manage to reduce the number of nodes to be processed during contract evaluation, by avoiding zero-value duplication.}
 
 \myrb{Let us now proceed to discuss the possible state transitions between time steps and their probabilities. First of all, we define $\left( \mu_r\right)_{i,k}=k_r\left(\theta-\mathcal{R}_{i,k} \right)$ as the drift coefficient at  $\mathcal{R}_{i,k}$. Then, for a node $\mathcal{R}_{i,k}$,  \cite{Appolloni2015} define the level of the upcoming two nodes as
{\small \begin{align}
 k_d^{\mathrm{ACZ}}(i,k)&=\max\left\lbrace k^{*}:0\leq k^* \leq k\ \mathrm{and}\ \mathcal{R}_{i,k}+\left( \mu_r\right)_{i,k}\Delta t\geq \mathcal{R}_{i+1,k^*}  \right\rbrace\cup\left\lbrace 0 \right\rbrace, \\
  k_u^{\mathrm{ACZ}}(i,k)&=\min\left\lbrace k^{*}:k+1\leq k^* \leq i+1\ \mathrm{and}\ \mathcal{R}_{i,k}+\left( \mu_r\right)_{i,k}\Delta t\leq \mathcal{R}_{i+1,k^*} \right\rbrace 	\cup\left\lbrace i+1 \right\rbrace.
\end{align}}
Here, for $k=k_{\min}(i),\dots,i$, we  set
{\small \begin{align}\label{10}
		k_d(i,k)&=\max\left\lbrace k_d^{\mathrm{ACZ}}(i,k),   k_{\mathrm{min}}(i+1) \right\rbrace, \\
		k_u(i,k)&=
		\begin{cases}
			k_u^{\mathrm{ACZ}}(i,k) & \mathrm{if}\ \mathcal{R}_{i,k}<\theta,\\
				k_d(i,k)+1& \mathrm{otherwise}.
			\end{cases}
\end{align}} 
Moreover, it is possible to prove (see Appendix \ref{appendix:ub}) that for each time step $i$, there exist an index, denoted by $k_{\max}(i)$ so that all nodes $\mathcal{R}_{i,k}$ with $k>k_{\max}(i)$ cannot be reached when starting from $\mathcal{R}_{0,0}=r_0$. Therefore, one can discard from the tree those nods.  Therefore, the only useful nodes for defining the tree are those that verify the relation $k_{\mathrm{min}}(i)\leq k\leq  k_{\mathrm{max}}(i)  $. We stress out that this observation improves the efficiency of the algorithm, as it reduces drastically   the computational cost, in particular when a high number $NT$ of time steps is employed.\\
 The transition probabilities among the nodes are defined to match the first order approximation of the first  moment of the CIR process. Starting from the node $(i,k)$, the probability that the process jumps to $(i+1,k_u(i,k))$ is defined as
 \begin{equation}
	p^{R}_{i,k}=\max\left\lbrace 0, \min\left\lbrace1, \frac{(\mu_r)_{i,k}h + r_{i,k} - r_{i+1,k_d(i,k)}}{r_{i+1,k_u(i,k)} - r_{i+1,k_d(i,k)}} \right\rbrace \right\rbrace. 
\end{equation}  
 Of course, the probability  that the process $\overline{r}$ jumps to $(i+1,k_d(i,k))$ is $1-p^{R}_{i,k}$.
} 
 \subsubsection{The tree for the account value}\label{ss:tree_for_A}

The second step of the Tree-LTC algorithm is to create a lattice to discretize the underlying, i.e. account value $A$. Specifically, we approximate the process $A$ in $[0,T]$  with a discrete time process $\bar{A}=\left\lbrace \bar{A}_{i}\right\rbrace_{i=0,\dots,NT}$, so that $\bar{A}_{i}$ approximates $A_{i\Delta t}$. In \cite{Appolloni2015}, this grid of values, generated from a uniform mesh of values for the log-price of the account value, is time-dependent: the number of nodes grows linearly with the number of time steps, as usual in any tree structure. In our case, this fact hampers the evaluation of the GLWB-LTC contract because, as payments are made, the account value may experience downward movements due to withdrawals, thus assuming values outside the mesh of nodes in the tree. Consequently, we prefer to discretize the account value by a complete grid of values, which does not change over time and which defines the support for a Markov chain. Specifically, we set two values, $ A_{\min}\approx 0$ and $ A_{\max}>>P$ ($P$ is the initial value for $A$), and create a uniform mesh between the logarithm these two values. Specifically, we set
\begin{align*}
	j_{\min}&=-\min\left\lbrace j^*\in\mathbb{Z}\ \mathrm{s.t.}\ P\cdot\exp\left(  j^*\cdot \sigma \sqrt{\Delta t}\right)\geq A_{\min} \right\rbrace+1,\\
	  j_{\max}&=\max\left\lbrace j^*\in\mathbb{Z}\ \mathrm{s.t.}\ P\cdot\exp\left(  j^*\cdot \sigma \sqrt{\Delta t}\right)\leq A_{\max} \right\rbrace+j_{\min},
\end{align*} 
and for $j=1,\dots , j_{\max}$, we define the node values as
$$ \mathcal{A}_j=P\cdot\exp\left(\left( j-j_{\min}\right)\sigma \sqrt{\Delta t}\right),  $$ 
	so that $\mathcal{A}_1\approx A_{\min}$, $\mathcal{A}_{j_{\min}}=P,$ and $\mathcal{A}_{j_{\max}}\approx A_{\max}$. Moreover, since the account value can also be empty, we include zero among the possible values  by setting
	$\mathcal{A}_0=0$.
	Finally, we define $\mathcal{G}_{A}=\left\lbrace \mathcal{A}_j, j=0,1,\dots,j_{\min},\dots,j_{\max} \right\rbrace$ as the set of the nodes of the lattice for $A$.
	
	\subsubsection{The joint distribution}
	The marginal transition probabilities for the lattice for $A$ are not defined directly. Instead, joint probabilities are defined for the pair $\left(\mathcal{A}_j,\mathcal{R}_{i,k} \right) $. Specifically, suppose that at the $i$-th time step the location of the couple $\left(\bar{A}_i,\bar{r}_i \right)$ is given by $\left(\bar{A}_i,\bar{r}_i \right)=\left(\mathcal{A}_{{j}},\mathcal{R}_{i,k} \right)$. We begin the definition of the transition probabilities by assuming $j>0$, so that $\mathcal{A}_{{j}}>0$.
	  We define 
	\begin{align*} 
		& j_d(i,j,k)=\max \left\{j^*\ \mathrm{s.t.}\ 1 \leq  j^* <  {j} \text { and } \mathcal{A}_{{j}}\cdot\left( 1+\mathcal{R}_{i,k}\Delta t\right)   \geq  \mathcal{A}_{j^*}\right\} \cup \left\{1\right\},  
		\\ 
		& j_u(i,j,k)=\min \left\{j^*\ \mathrm{s.t.}\ {j} <  j^* \leq  j_{\max} \text { and } \mathcal{A}_{{j}}\cdot\left( 1+\mathcal{R}_{i,k}\Delta t\right)   \leq  \mathcal{A}_{j^*}\right\}\cup \left\{j_{\max}\right\}.
	\end{align*}
	\myrb{Moreover, it is possible to prove that, as $\Delta t$ tends to zero, $j_d(i,{j}, {k})$ and $j_u(i,{j}, {k})$ converge respectively to $j$ and $j+1$ for all $j=2,\dots,j_{\max}-1$ (see Appendix \ref{appendix:co}).}
 	The probability of an up-movement of the tree for $\bar{A}$ is set as:
	$$ p^{A}_{i,j,k}=\max\left\lbrace  \min \left\lbrace \frac{\mathcal{A}_{{j}}\cdot\left( 1+\mathcal{R}_{i,k}\Delta t\right) -\mathcal{A}_{i_d(i,{j}, {k})}}{ \mathcal{A}_{j_u(i,{j}, {k})} -\mathcal{A}_{j_d(i,{j}, {k})}   },1 \right\rbrace ,0 \right\rbrace.$$
  To simplify notation, we write $j_d$, $j_u$ and $p^{A}_u$ instead of $  j_d(i,{j}, {k}),j_u(i,{j}, {k})$ and $ p^{A}_{i,j,k}$ respectively,  leaving the dependence on ${i,j}$ and ${k}$, taking it as for granted.  Moreover, let $p^{A}_d=1-p^{A}_u$ the probability for a down movement.
	
	\myrb{Now, let us denote with $k_d$ and $k_u$ the position of future nodes from $\mathcal{R}_{i,k}$ (also in this case we omit the dependence on $i $ and ${k}$), and let  $p^{R}_{d}$ and $p^{R}_{u}$ be the  probabilities for a down and an up movement of the process $\overline{r}$, respectively.}
	
	\myrb{Starting from an assigned node, the discrete time processes $\bar{A}$ can move to two future nodes, and so does the process $\bar{r}$. Thus, the future nodes associated with the pair $\left(\mathcal{A}_j,\mathcal{R}_{i,k} \right) $ are four, namely:
	$$\left(\mathcal{A}_{j_d},\mathcal{R}_{i+1,k_d} \right),\  \left(\mathcal{A}_{j_d},\mathcal{R}_{i+1,k_u} \right),\  \left(\mathcal{A}_{j_u},\mathcal{R}_{i+1,k_d} \right),\    \left(\mathcal{A}_{j_u},\mathcal{R}_{i+1,k_u} \right), $$
	and let 
	$$ p_{d,d}, \ p_{d,u},\ p_{u,d}, \ p_{u,u}, $$ be the corresponding probabilities.	
	These probabilities are determined as the unique solution of the following linear system, whose equations correspond to imposing the matching of the first two moments for both the processes $r$ and $A$, and of the covariance between the two processes:
	\begin{equation}\label{sistema}
			\begin{cases}
			p_{d,d}+p_{d,u}+p_{u,d}+ p_{u,u}&=1,\\
			p_{d,d}+p_{d,u} &=p^{A}_{d},\\
			p_{d,d}+p_{u,d} &=p^{R}_{d},\\ 
			m_{d,d}p_{d,d}+ m_{d,u}p_{d,u}+m_{u,d}p_{u,d} + m_{u,u}p_{u,u}&=\rho \sigma_r\sigma_F\sqrt{\mathcal{R}_{i,k}}\mathcal{A}_{j}\Delta t,\\
		\end{cases}
	\end{equation}
with 
$$m_{d,d}=\left(\mathcal{A}_{j_d}\!-\!\mathcal{A}_{{j}} \right)\!\left(\mathcal{R}_{i+1,k_d}\!-\!\mathcal{R}_{i,k} \right),\quad
m_{u,d}=\left(\mathcal{A}_{j_u}\!-\!\mathcal{A}_{{j}} \right)\!\left(\mathcal{R}_{i+1,k_d}\!-\!\mathcal{R}_{i,k} \right),  
$$ 
$$m_{d,u}=\left(\mathcal{A}_{j_d}\!-\!\mathcal{A}_{{j}} \right)\!\left(\mathcal{R}_{i+1,k_u}\!-\!\mathcal{R}_{i,k} \right),\quad
m_{u,u}=\left(\mathcal{A}_{j_u}\!-\!\mathcal{A}_{{j}} \right)\!\left(\mathcal{R}_{i+1,k_u}\!-\!\mathcal{R}_{i,k} \right).
$$
This system always admits one and only one positive solution, as discussed in   Appendix \ref{appendix:co}.
 Finally, we discuss the case $i=0$, which corresponds to  $\mathcal{A}_{0}=0$. The value $0$ is an absorbing class for the account value: once $A$ is depleted, it can no longer become positive again. Therefore, if $\left(\bar{A}_i,\bar{r}_i \right)=\left(\mathcal{A}_0,\mathcal{R}_{i,k} \right)$, then the nodes reachable by the process at the next instant are 	$\left(\mathcal{A}_{0},\mathcal{R}_{i+1,k_d} \right)$ and $ \left(\mathcal{A}_{0},\mathcal{R}_{i+1,k_u} \right)$, with probabilities equal to $p_d^R$ and $p_u^R$ respectively. }
	
	\subsection{Pricing}
 We apply the Tree-LTC method to compute an approximation $\bar{\mathcal{V}}$ of the GLWB-LTC contract value $\mathcal{V}$.  First of all,  we set $N$, the number of time steps per year (as defined in Subsection \ref{ss:tree_for_r}), $A_{\min}$ and $A_{\max}$, the limits for the positive nodes of $\mathcal{G}_A$ (as defined in Subsection \ref{ss:tree_for_A}). In addition, we recall that $T_{}$ is the difference between $122$ (maximum age) and the initial age $x_0$ of the insured.
 At each time step $i$ of the Tree-LTC algorithm, we define a grid of values to diffuse the processes $\bar{A},\bar{r}$ and $M$:
\begin{equation}\label{ggg}
	  \mathcal{G}_i=\mathcal{G}_A\times \mathcal{G}^i_r\times\left\lbrace {1,\dots,7}\right\rbrace, 
\end{equation} where the set $\left\lbrace {1,\dots,7}\right\rbrace$ describes the health states of the PH.
For each anniversary $n=0,\dots, T$, we define a function $\bar{\mathcal{V}}_{n}$ which approximates the real contract fair value at the year $n$. Specifically, for any point $\left(\mathcal{A}_{j}, \mathcal{R}_{nN,k}, k \right) \in \mathcal{G}$ we have
$$\bar{\mathcal{V}}_n\left(\mathcal{A}_j, \mathcal{R}_{nN,k}, h \right)\approx\mathcal{V}\left(\mathcal{A}_j, \mathcal{R}_{nN,k}, h,n \right).$$
We stress out that, by similarity reduction discussed in Section \ref{sr}, we can assume $B_n=P$ for all anniversaries, so, hereinafter, $G_n$ and $L_{n}\left(M_n \right) $ are computed according to $B_{n}^{1+}=B_{n}^{2+}=P$.
The computation of the function $\bar{\mathcal{V}}_n$ is achieved by proceeding backward in time.
At maturity, i.e. $n=T$ and $i=NT $, no PH is longer alive, so we set:
$$\bar{\mathcal{V}}_{T }\left(\mathcal{A}_j, \mathcal{R}_{NT,k}, h \right)=
G_{T }+\max{\left( 0,\mathcal{A}_j-G_{T }\right)}.$$ 

Let us now consider a general anniversary $n\in\left\lbrace 0,\dots, T -1\right\rbrace $ and assume that we have already calculated the function $\bar{\mathcal{V}}_{n+1}$ at the anniversary $n+1$. To compute $\bar{\mathcal{V}}_{n}$  on the grid $\mathcal{G}_{nN}$  for ${h}= 7$, just set
$$\bar{\mathcal{V}}_{n}\left( \mathcal{A}_j, \mathcal{R}_{nN,k},  7\right) = G_{n}+\max{\left( 0,\mathcal{A}_j-G_{n}\right). }$$

As far as the health state $h\neq 7$ is considered, the following actions are carried out in this specified order.
\begin{enumerate}
	\item Mix the values of $\bar{\mathcal{V}}_{n+1}$ according to the health transition probability $p_{h,h'}\left(n,n+1 \right)$ (from state $h$ at year $n$ to state $h'$ at year $n+1$). Specifically, we define:
	$$\bar{\mathcal{V}}_{n+1}^{\mathrm{mix}}\left(\mathcal{A}_j, \mathcal{R}_{(n+1)N,k},  h \right)=\sum_{h'=1}^{7}p_{h,h'}\left(n,n+1 \right)\bar{\mathcal{V}}_{n+1}^{\mathrm{}}\left(\mathcal{A}_j, \mathcal{R}_{(n+1)N,k}, h' \right).$$
	\item Compute the discount expected value of the mix, by using the Tree-LTC algorithm. Specifically, we divide the time lapse $[n,n+1]$ into $N$ sub-intervals. Let us term:
	$$\bar{\mathcal{V}}_{n,N}\left( \mathcal{A}_j, \mathcal{R}_{(n+1)N,k}, h\right)=\bar{\mathcal{V}}_{n+1}^{\mathrm{mix}}\left( \mathcal{A}_j, \mathcal{R}_{(n+1)N,k}, h\right),$$ 
	as the contract value at time $n+1$ before any payment is performed. For each sub-time step $i=(n+1)N-1,\dots ,nN$ we employ the Tree-LTC algorithm. We distinguish some cases. 
	\begin{enumerate}
		\item \myrb{If $j=0$, that is $\mathcal{A}_{j}=0$, then
	\begin{samepage}		
	{\small\begin{multline*}
		 \bar{\mathcal{V}}_{n,i}\left(\mathcal{A}_0, \mathcal{R}_{i,k},h\right) =e^{-\Delta t \mathcal{R}_{i,k} }
		 \left[\phantom{+}p^{R}_{d}\bar{\mathcal{V}}_{n,i+1}\left( \mathcal{A}_{0 }, \mathcal{R}_{i+1,k_d\left(k \right) }, h \right)\right.\\
		 \left.+p_{u}^{R}\bar{\mathcal{V}}_{n,i+1}\left( \mathcal{A}_{0 }, \mathcal{R}_{i+1,k_u\left(k \right) }, h \right)\right].
	\end{multline*}}
\end{samepage}}
	\item \myrb{If $j=2,\dots,j_{\max}-1$,
	\begin{samepage}	
	{\small	\begin{align*}
			\bar{\mathcal{V}}_{n,i}\left(\mathcal{A}_j, \mathcal{R}_{i,k},h\right) =e^{-\Delta t \mathcal{R}_{i,k} }\left[\right.
			&\left.\phantom{+}p_{d,d}\bar{\mathcal{V}}_{n,i+1}\left( \mathcal{A}_{j_d\left(j,k \right) }, \mathcal{R}_{i+1,k_d\left(k \right) }, h \right)\right.  \\
			&
			+p_{d,u}\bar{\mathcal{V}}_{n,i+1}\left( \mathcal{A}_{j_d\left(j,k \right) }, \mathcal{R}_{i+1,k_u\left(k \right) }, h \right)  \\
			&
			+p_{u,d}\bar{\mathcal{V}}_{n,i+1}\left( \mathcal{A}_{j_u\left(j,k \right) }, \mathcal{R}_{i+1,k_d\left(k \right) }, h \right)  \\
			&\left. 
			+p_{u,u}\bar{\mathcal{V}}_{n,i+1}\left( \mathcal{A}_{j_u\left(j,k \right) }, \mathcal{R}_{i+1,k_u\left(k \right) }, h \right)\right].
	\end{align*}}
\end{samepage}
}
\item\myrb{ If $j=1$ or $j=j_{\max}$,   we use linear interpolation to estimate $\bar{\mathcal{V}}_{n,i}\left(\mathcal{A}_j, \mathcal{R}_{i,k},h\right) $. That is because   the points $\mathcal{A}_1 e^{-\sigma_F\sqrt{\Delta t}}$ and $\mathcal{A}_{j_{\max}} e^{\sigma_F\sqrt{\Delta t}}$ are not included in the grid $\mathcal{G}_A$, and therefore it is necessary to impose some boundary conditions to determine the value of the contract at these points. This condition can be justified by the fact that when the account value is very large or very small, the contract value tends to behave as a linear function of the account value itself, as already remarked and exploited by \cite{Forsyth2014}.
	Specifically, we set
{\footnotesize 	\begin{align*}
	\bar{\mathcal{V}}_{n,i}\left(\mathcal{A}_1, \mathcal{R}_{i,k},h\right) &=\frac{	\bar{\mathcal{V}}_{n,i}\left(\mathcal{A}_3, \mathcal{R}_{i,k},h\right)-	\bar{\mathcal{V}}_{n,i}\left(\mathcal{A}_2, \mathcal{R}_{i,k},h\right)  }{\mathcal{A}_3-\mathcal{A}_2}\left(\mathcal{A}_1-\mathcal{A}_2 \right) +\bar{\mathcal{V}}_{n,i}\left(\mathcal{A}_2, \mathcal{R}_{i,k},h\right), \\
	   \bar{\mathcal{V}}_{n,i}\left(\mathcal{A}_{j_{\max}}, \mathcal{R}_{i,k},h\right) &=\frac{	\bar{\mathcal{V}}_{n,i}\left(\mathcal{A}_{j_{\max}-2}, \mathcal{R}_{i,k},h\right)-	\bar{\mathcal{V}}_{n,i}\left(\mathcal{A}_{j_{\max}-1}, \mathcal{R}_{i,k},h\right)  }{\mathcal{A}_{j_{\max}-2}-\mathcal{A}_{j_{\max}-1}}\left(\mathcal{A}_{j_{\max}}-\mathcal{A}_{j_{\max}-1} \right)\\& +\bar{\mathcal{V}}_{n,i}\left(\mathcal{A}_{j_{\max}-1}, \mathcal{R}_{i,k},h\right). \\
\end{align*}}}
\end{enumerate}\myrb{Moreover, if the full dynamic approach is considered, at each sub-time step $i$, we replace $\bar{\mathcal{V}}_{n,i}\left(\mathcal{A}_j, \mathcal{R}_{i,k},h\right)$  with 
\[\max\left\lbrace \mathcal{A}_{j}\left(1-\kappa_{n} \right) , \bar{\mathcal{V}}_{n,i}\left(\mathcal{A}_j, \mathcal{R}_{i,k},h\right)\right\rbrace \] to account for the possibility of a total surrender at time $t=i\Delta t$.}
\item Account for the possible withdrawal (only if $n>0$).  We  term 
$ \bar{\mathcal{V}}^{3+}_{n}= \bar{\mathcal{V}}_{n,nN}$ the contract value at anniversary $n$ after all payments are performed. Let $\gamma=\gamma_{n}\left(\mathcal{A}_j, \mathcal{R}_{nN,k}, {h} \right) $ be the value determined according to the withdrawal strategy considered, for the withdrawal at the $n$-th anniversary, for $A_n=\mathcal{A}_j$, $B_n=P$, $r=\mathcal{R}_{nN,k}$ and $M_n=h$. The contract value before the withdrawal takes place, denoted by $\bar{\mathcal{V}}^{2+}_{n}$ is computed as follows.   So
\begin{itemize}
	\item If $\gamma=0$, then 
\begin{equation}\label{eq:Ngamma0}
	 \bar{\mathcal{V}}^{2+}_{n}\left(\mathcal{A}_j, \mathcal{R}_{nN,k}, {h} \right)=
	 \left( 1+b\right) \bar{\mathcal{V}}^{3+}_{n}\left(\frac{\mathcal{A}_j}{1+b}, \mathcal{R}_{nN,k}, {h} \right).
\end{equation}
	\item If $0<\gamma\leq1$, then 
	\begin{equation}\label{eq:Ngamma01}
		\bar{\mathcal{V}}^{2+}_{n}\left(\mathcal{A}_{j},\mathcal{R}_{nN,k}, {h} \right)=
	\bar{\mathcal{V}}^{3+}_{n}\left(\max\left( \mathcal{A}_{j}-W_{n},0\right) ,\mathcal{R}_{nN,k}, {h} \right)+Y_{n}, \end{equation}
	with $W_{n}$ and $Y_{n}$ as in \eqref{eq:gamma01}.
	\item If $1<\gamma\leq2$, then 
	\begin{equation}\label{eq:Ngamma12}
		\bar{\mathcal{V}}^{2+}_{n}\left(\mathcal{A}_{j},\mathcal{R}_{nN,k}, {h} \right)=\left(2-\gamma \right) 
	\bar{\mathcal{V}}^{3+}_{n}\left(\frac{\max\left( \mathcal{A}_{j}-W_{n},0\right)}{\left(2-\gamma \right)},\mathcal{R}_{nN,k}, {h} \right)+Y_{n},
\end{equation}
	with $W_{n}$ and $Y_{n}$ as in \eqref{eq:gamma12}.
\end{itemize}
We stress out that in equations \eqref{eq:Ngamma0}, \eqref{eq:Ngamma01} and \eqref{eq:Ngamma12},  the post-withdrawal value of $A$ may  not be in the grid $\mathcal{G}_{A}$. In this case,   interpolation is used to compute $	\bar{\mathcal{V}}^{3+}_{n}$.
\item Pay the LTC (only if $n>0$). We term 
$ \bar{\mathcal{V}}^{1+}_{n}$ the contract value at anniversary $n$ before the payment of the LTC. Then
 $$ \bar{\mathcal{V}}^{1+}_{n}\left(\mathcal{A}_{j},\mathcal{R}_{nN,k}, {h} \right)=
  \bar{\mathcal{V}}^{2+}_{n}\left(\max \left( \mathcal{A}_{j}- L_{n}(h),0\right),\mathcal{R}_{nN,k}, {h} \right)+L_{n}(h).
  $$
\item Fees adjustment. We term 
  $ \bar{\mathcal{V}}^{-}_{n}$ the contract value at anniversary $n$ before fees are withdrawn. Then
  $$ \bar{\mathcal{V}}^{-}_{n}\left(\mathcal{A}_{j},\mathcal{R}_{nN,k}, {h} \right)=
  \bar{\mathcal{V}}^{1+}_{n}\left(\max \left( \mathcal{A}_{j}(1-\alpha)-\beta P,0\right),\mathcal{R}_{nN,k}, {h} \right).
  $$
\end{enumerate}
 Through the above equations, by moving backward in time, it is possible to calculate the price of the contract up to the initial time $t=0$. Specifically, the initial price $\mathcal{V}\left(P,P,r_0,M_0,0^- \right)$ is approximated by  $\bar{\mathcal{V}}_{0}^{-}\left(\mathcal{A}_{j_{\min}},\mathcal{R}_{0,0},M_0  \right)$.
  
 \begin{remark}
 	In the case of dynamic withdrawal, the identification of the optimal withdrawal strategy can be done by comparing, for different $\gamma$ values on a mesh from $\gamma=0$ to $\gamma=2$, the one that maximizes the overall value of the contract $\bar{\mathcal{V}}_n^{2+}$, which can be computed by equations $\eqref{eq:Ngamma0}$, $\eqref{eq:Ngamma01}$ or $\eqref{eq:Ngamma12}$.
 \end{remark}

	 \begin{remark}
	 	The previously described procedure, which is valid for the stochastic BS-CIR model, can be easily readapted to the Black and Scholes sub-model. Indeed, it will be sufficient to assume a constant interest rate.
	 \end{remark}
 
 \begin{remark}\label{Rml_alpha}
 	A common practice in the field of variable annuities is to calculate the value of the $\alpha$ parameter that makes the contract fair, that is $\bar{\mathcal{V}}_{0}^{-}\left(\mathcal{A}_{j_{\min}},\mathcal{R}_{0,0},M_0  \right)=P$. This calculation can be done easily by iterating the initial price calculation for different values of $\alpha$, based on an appropriate zero-search scheme, such as the secant method we employed.
 	 \end{remark}

\begin{remark}
\myrb{	The proposed evaluation technique has several advantages in itself: the number of nodes used at each time-step to discretise the continuous processes is bounded. In addition, the interpolation technique allows for simple and efficient handling of jumps in the account value due to withdrawal payments. Furthermore, various exercise strategies, such as the dynamic strategy, can be handled with very little computational effort. }
 \end{remark}

	\section{Numerical results}\label{NR}
	In this Section, we present the results of some numerical tests in which we test the evaluation procedure based on the Tree-LTC algorithm. Specifically, we calculate the fair value of the fee parameter $\alpha$ (see Remark \ref{Rml_alpha}) as certain parameters change, such as, for example, the age $x_0$ of the PH at inception or the withdrawal strategy.
	Under the static withdrawal strategy, we compare the fair fee $\alpha$ arising from the Tree-LTC against the ones obtained by implementation of a classical Monte Carlo method.
	Such tests are performed under the assumption of stochastic interest rates, but also within the Black-Scholes model framework. In this simpler setting, indeed, we are able to compare more closely the performance of our algorithm against the performance of the Monte Carlo method with control variates that is used by \cite{Hsieh2018} for the pricing of the considered contract. 
	
	In Table \ref{tab:param1}, we report a brief description of the parameters that characterize the contract and the underlying stochastic models, along with the respective symbols and the values we assigned to them within our numerical experiments.  We also point out that, for the tests performed with the Black-Scholes model, we assume $r=r_0=5\%$. Furthermore, the guaranteed minimum withdrawal rate $g$ is indexed to the age of the PH: $3\%$ for a 60-year-old with an increase of $0.1\%$ for each additional year of age. This choice is made to make the value of the contract more homogeneous: a person aged 80 has a shorter life expectancy than one aged 60, so to make the contract more equitable we need to increase the guaranteed amount for withdrawals. 
	
	\myrb{As far as the BS-CIR model is considered, we assume a negative value for the correlation parameter $\rho$, specifically $\rho=-0.25$. In fact, in financial markets, the relationship between stock market performance and interest rates can vary, but typically exhibits a negative correlation. However, in the following we will analyse different  values for $\rho$ and their effects on the value of the contract (see Figure \ref{fig:2}).}
	
	Before presenting the numerical results, we point out that both the Tree-LTC algorithm and the Monte Carlo methods are implemented in the C language and were run on the same machine (i5-1035G1 CPU processor, 8 GB RAM) in order to compare computation times.
	
 \begin{table}[h]
 	\begin{center}
 		\scalebox{0.95}{ %
 			\begin{tabular}{cllccll}
 				\toprule 
 				Symbol & Name & Value &  & Symbol & Name & Value\tabularnewline
 				\midrule
 				$P$ & initial account value & $100$ &  & $\alpha$ & fees proportional to $A$ & variable\tabularnewline
 				$\sigma_{F}$ & volatility of the fund & \myrb{$0.20$} &  & $\beta$ & fees proportional to $B$ & $0.003$\tabularnewline
 				$\sigma_{r}$ & volatility of the i.r. & \myrb{$0.10$} &  & $x_{0}$ & entry age & $60,65,70,75 $ or $80$\tabularnewline
 				$k_{r}$ & speed of mean rev & $0.5$ &  & $g$ & withdrawal rate & $0.03\!+\!(x_{0}\!-\!60)\!\cdot\!0.001$\tabularnewline
 				$\theta$ & long mean i.r. & $0.05$ &  & $c$ & LTC withdrawal rate & $0.06$ or $0.00$\tabularnewline
 				$r_{0}$ & initial interest rate & $0.05$ &  & $\pi$ & inflation protection & $0.05$\tabularnewline
 				$\rho$ & correlation & $-0.25$ &  & $b$ & bonus & $g+0.005$\tabularnewline
 				$M_{0}\left(0\right)$ & initial health state & $1$ &  & $\kappa_{n}$ & penalty for $\gamma_{n}>1$ & $0.01\cdot\max(0,8-t)$\tabularnewline
 				\bottomrule
 			\end{tabular}
 			
 	}\end{center}
 	
 	\caption{\label{tab:param1}Parameters employed for numerical experiments.}
 \end{table}

 \subsection{The Black-Scholes model}
 
 In this Subsection, we work in the framework of the Black-Scholes model for the description of the dynamics of the underlying fund, without any stochastic assumption about the interest rate. Indeed, \cite{Hsieh2018} consider a non-stochastic interest rate and a static withdrawal strategy and their   approach  to the contract evaluation relies on the Monte Carlo method with control variates (MC-CV). Accordingly, choosing the most simple setting for the interest rate and the withdrawal strategy allows us to preliminarly validate the Tree-LTC, by comparing its pricing performance against the one of the Monte Carlo and the MC-CV methods. Specifically, MC-CV is a Monte Carlo algorithm that exploits the following four control variates to reduce the variance of the results:
\begin{equation}\label{eq:CV}
	\begin{aligned}
	C_1&=A_\tau^{-}e^{-r\tau}-\mathbb{E}^\mathbb{Q}\left[ A_\tau^{-}e^{-r\tau}\right], \\
	C_2&=F_\tau-\mathbb{E}^\mathbb{Q}\left[F_\tau\right],\\
	C_3&=\sum_{n=1}^\tau\left(G_n+L_{n}\left(M(n)\right)\right)-\mathbb{E}^\mathbb{Q}\left[ \sum_{n=1}^\tau\left(G_t+L_{n}\left(M(n)\right)\right)\right],\\
	C_4&=\tau-\mathbb{E}^\mathbb{Q}\left[\tau\right].
\end{aligned}
\end{equation}
Please, observe that  $\tau$ is the anniversary year  immediately after   the PH's death.

First of all, we test the convergence of the three considered algorithms by changing the number of discretization steps. In particular, we consider four parameter configurations, denoted by the letters A, B, C and D, as shown in Table \ref{tab:BS0}. In particular, as far as the Monte Carlo algorithms are considered we report the number of simulations and the number of time discretization steps (in the Black-Scholes model, exact simulation is possible, so we always consider only one step per year). As far as the Tree-LTC algorithm is employed, we report first the number $N$ of time steps per year, and then the factor $f_A$ which is used to compute $A_{\min}$ and $A_{\max}$ as $A_{\min}=P/f_A$ and $A_{\max}=P\cdot f_A$. 

Convergence results are displayed in Table \ref{tab:BS1}. Specifically, we compute the fair value of $\alpha$ for a PH with entry age $x_0=60$, by changing the parameter setup. Moreover, we consider both  a GLWB product that includes a LTC guarantee amounting to 6\% of the inflation-indexed benefit base, i.e., $c=0.06$, and, for comparison purposes, a GLWB product that does not include a LTC guarantee, i.e., $c=0$ (traditional GLWB annuity). \myrb{The findings outlined in Table \ref{tab:BS1} indicate that the point estimate of \(\alpha\) obtained via the Tree-LTC method lies within the confidence intervals established by the first two Monte Carlo methods. This consistency underscores the compatibility of the three numerical techniques in determining the estimated values of \(\alpha\).
} The MC-CV method turns out to be more effective than the MC method: the confidence interval amplitudes are smaller for approximately the same computational time. The results produced by the Tree-LTC method are much more stable than the results related to the other two methods, and the computational times are significantly shorter.
For all the considered numerical methods, setup D, the most accurate by far, was also used in the other numerical tests, presented below.

In Table \ref{tab:BS2}, we show the fair values of $\alpha$ at different entry ages for the PH, being five years apart. Also in this case, we consider both the case where the LTC payment is provided and the case where no LTC benefit is granted. We see that the outcomes of the tree numerical methods point to the same pattern of the fair value of $\alpha$ as age increases. Furthermore, we remark that embedding the LTC component increases the fair value of $\alpha$, but to a small extent, never exceeding $120$ basis points of the value of the traditional GLWB annuity. 
In Table \ref{tab:BS3}, we show the fair values of $\alpha$ under the same BS assumption for the dynamics of the mutual fund, but under different cases for the withdrawal strategy, either static, or mixed or dynamic or full dynamic. This implies the exclusive use of the Tree-LTC algorithm. Indeed,  according to the previous evidence,  the proposed algorithm turns out  to be the most accurate and the fastest  among the competing methods. Furthermore, among the considered numerical methods, the Tree-LTC is the only one able to tackle, in a straightforward way, the stochastic control problem involved by the the dynamic withdrawal strategy.
The age being fixed, the more numerous the withdrawal options for the PH the higher the fair value of the fee $\alpha$. Nevertheless, such an increase in $\alpha$ appears modest and the total cost never exceeds $250$ basis points.

  To conclude this battery of tests, we investigate what impact the LTC and the withdrawal strategy have on the initial contract price. Table \ref{tab:BS4} shows the prices, calculated using the Tree-LTC method, of the GLWB-LTC contract as the age of the PH, the withdrawal strategy and the amount of the LTC change. In the cases considered here, for each value of $x_0$, $\alpha$ is set equal to the fair value in the case of static withdrawal, for $c=0\%$. For this reason, the price in the seventh column of Table \ref{tab:BS4} is always $100.00$. More generally, the prices for $c=0$ are all close to $100$. We then observe that the values for all the strategies with respect to $c=6\%$ are greater than $100$, as to be expected, but never exceed $10$ monetary units with respect to the relative cases for $c=0$.  This cost is not very large if one takes into account that LTC significantly increases the guaranteed minimum payment in the case of disability. This small price difference may be attractive to buyers, incentivizing them to purchase policies with LTC.

 \begin{table}[h] 
 	\begin{center}
 		\scalebox{1}{ 
 			
 			\begin{tabular}{ccccc}
 				\toprule 
 				Setup &  & MC & MC-CV & Tree-LTC\tabularnewline
 				\midrule
 				{\small{}$A$} &  & $1\cdot10^{6},1$ & $1\cdot10^{6},1$ & $100,100$\tabularnewline
 				{\small{}$B$} &  & $2\cdot10^{6},1$ & $2\cdot10^{6},1$ & $200,200$\tabularnewline
 				{\small{}$C$} &  & $4\cdot10^{6},1$ & $4\cdot10^{6},1$ & $400,400$\tabularnewline
 				{\small{}$D$} &  & $8\cdot10^{6},1$ & $8\cdot10^{6},1$ & $800,800$\tabularnewline
 				\bottomrule
 			\end{tabular}
 			
 	}\end{center}
 	
 	\caption{\label{tab:BS0} Parameter setup for the numerical algorithms when the Black-Scholes model is considered.}
 \end{table}
 \begin{table}[h]
 	\begin{center}
 		\scalebox{0.95}{ 
\begin{tabular}{ccccccccc}
	\toprule 
	  &  & \multicolumn{3}{c}{$c=0.06$} &  & \multicolumn{3}{c}{$c=0$}\tabularnewline
	\cmidrule{3-9} \cmidrule{4-9} \cmidrule{5-9} \cmidrule{6-9} \cmidrule{7-9} \cmidrule{8-9} \cmidrule{9-9} 
	Setup &  & MC & MC-CV & Tree-LTC &  & MC & MC-CV & Tree-LTC\tabularnewline
	\midrule
	{\small{}$A$} &  & $\underset{\left(34\right)}{154.70\pm1.60}$ & $\underset{\left(31\right)}{154.71\pm0.66}$ & $\underset{\left(0.3\right)}{154.37}$ &  & $\underset{\left(27\right)}{55.07\pm1.35}$ & $\underset{\left(27\right)}{54.95\pm0.45}$ & $\underset{\left(0.2\right)}{54.62}$\tabularnewline
	{\small{}$B$} &  & $\underset{\left(53\right)}{154.76\pm1.13}$ & $\underset{\left(66\right)}{154.36\pm0.46}$ & $\underset{\left(0.6\right)}{154.44}$ &  & $\underset{\left(61\right)}{55.27\pm0.95}$ & $\underset{\left(53\right)}{54.67\pm0.32}$ & $\underset{\left(0.4\right)}{54.76}$\tabularnewline
	{\small{}$C$} &  & $\underset{\left(128\right)}{154.84\pm0.80}$ & $\underset{\left(111\right)}{154.30\pm0.33}$ & $\underset{\left(0.8\right)}{154.46}$ &  & $\underset{\left(113\right)}{55.25\pm0.67}$ & $\underset{\left(107\right)}{54.64\pm0.22}$ & $\underset{\left(1.0\right)}{54.80}$\tabularnewline
	{\small{}$D$} &  & $\underset{\left(212\right)}{154.83\pm0.56}$ & $\underset{\left(213\right)}{154.49\pm0.23}$ & $\underset{\left(1.7\right)}{154.47}$ &  & $\underset{\left(229\right)}{55.12\pm0.48}$ & $\underset{\left(224\right)}{54.76\pm0.16}$ & $\underset{\left(1.8\right)}{54.81}$\tabularnewline
	\bottomrule
\end{tabular}
 			
 	}\end{center}
 	
 	\caption{\label{tab:BS1}The fair values of $\alpha$ (in basis points), in the Black-Scholes model, by changing the numerical setup and by assuming  the presence ($c=6\%$) or absence ($c=0$) of LTC. The values in parentheses indicate computational time in seconds.}
 \end{table}
\begin{table}[h]
	\begin{center}
		\scalebox{0.95}{ %
\begin{tabular}{ccccccccc}
	\toprule 
	Entry &   & \multicolumn{3}{c}{$c=0.06$} &  & \multicolumn{3}{c}{$c=0$}\tabularnewline
	\cmidrule{3-9} \cmidrule{4-9} \cmidrule{5-9} \cmidrule{6-9} \cmidrule{7-9} \cmidrule{8-9} \cmidrule{9-9} 
	age &  & MC & MC-CV & Tree-LTC &  & MC & MC-CV & Tree-LTC\tabularnewline
	\midrule
	{\small{}$60$} &  & $\underset{\left(228\right)}{154.65\pm0.56}$ & $\underset{\left(233\right)}{154.49\pm0.23}$ & $\underset{\left(0.8\right)}{154.46}$ &  & $\underset{\left(229\right)}{55.12\pm0.48}$ & $\underset{\left(224\right)}{54.76\pm0.16}$ & $\underset{\left(0.9\right)}{54.80}$\tabularnewline
	{\small{}$65$} &  & $\underset{\left(165\right)}{166.68\pm0.61}$ & $\underset{\left(170\right)}{166.77\pm0.25}$ & $\underset{\left(0.9\right)}{166.86}$ &  & $\underset{\left(155\right)}{55.18\pm0.50}$ & $\underset{\left(182\right)}{55.28\pm0.17}$ & $\underset{\left(0.7\right)}{55.36}$\tabularnewline
	{\small{}$70$} &  & $\underset{\left(144\right)}{166.55\pm0.63}$ & $\underset{\left(174\right)}{166.73\pm0.25}$ & $\underset{\left(0.8\right)}{166.80}$ &  & $\underset{\left(158\right)}{48.89\pm0.52}$ & $\underset{\left(148\right)}{49.13\pm0.17}$ & $\underset{\left(0.6\right)}{49.13}$\tabularnewline
	{\small{}$75$} &  & $\underset{\left(128\right)}{156.88\pm0.65}$ & $\underset{\left(150\right)}{156.94\pm0.24}$ & $\underset{\left(0.6\right)}{156.93}$ &  & $\underset{\left(138\right)}{38.14\pm0.54}$ & $\underset{\left(133\right)}{38.25\pm0.15}$ & $\underset{\left(0.5\right)}{38.24}$\tabularnewline
	{\small{}$80$} &  & $\underset{\left(112\right)}{140.67\pm0.67}$ & $\underset{\left(175\right)}{140.40\pm0.22}$ & $\underset{\left(0.6\right)}{140.27}$ &  & $\underset{\left(130\right)}{25.11\pm0.56}$ & $\underset{\left(124\right)}{25.12\pm0.12}$ & $\underset{\left(0.5\right)}{25.04}$\tabularnewline
	\bottomrule
\end{tabular}

	}\end{center}
	
	\caption{\label{tab:BS2}The fair values of $\alpha$ (in basis points), in the Black-Scholes model, by changing the entry age of the PH and by assuming  the presence ($c=6\%$) or absence ($c=0$) of LTC. The values in parentheses indicate computational time in seconds.}
\end{table}
\begin{table}[h]
	\begin{center}
		\scalebox{1}{ 
			
\begin{tabular}{ccccccccccc} 
	\toprule 
	 &  & \multicolumn{4}{c}{$c=0.06$} & \multicolumn{1}{c}{} & \multicolumn{4}{c}{$c=0$}\tabularnewline
	Entry &  & \multicolumn{4}{c}{strategy:} &  & \multicolumn{4}{c}{strategy:}\tabularnewline
	age &  & static & mixed & dynamic & full dyn &  & static & mixed & dynamic & full dyn\tabularnewline
	\midrule
	{\small{}$60$} &  & $154.46$ & $217.02$ & $229.62$ & $244.55$ &  & $54.80$ & $82.14$ & $85.74$ & $88.06$\tabularnewline
	{\small{}$65$} &  & $166.86$ & $211.19$ & $222.85$ & $233.86$ &  & $55.36$ & $74.64$ & $77.63$ & $79.33$\tabularnewline
	{\small{}$70$} &  & $166.80$ & $195.30$ & $205.67$ & $212.76$ &  & $49.13$ & $61.33$ & $63.62$ & $64.76$\tabularnewline
	{\small{}$75$} &  & $156.93$ & $173.33$ & $182.20$ & $186.09$ &  & $38.24$ & $45.06$ & $46.70$ & $47.38$\tabularnewline
	{\small{}$80$} &  & $140.27$ & $148.44$ & $155.74$ & $157.49$ &  & $25.04$ & $28.28$ & $29.38$ & $29.74$\tabularnewline
	\bottomrule
\end{tabular}
			
	}\end{center}
	
	\caption{\label{tab:BS3}The fair values of $\alpha$ (in basis points), in the Black-Scholes model, by changing
		the entry age of the PH, the presence ($c=6\%$) or absence ($c=0$)
		of LTC, and the withdrawal strategy (static, mixed, dynamic or full dynamic). The values in parentheses indicate
		computational time in seconds.}
\end{table}

\begin{table}[h]
	\begin{center}
		\resizebox{\textwidth}{!}{ 
\setlength{\tabcolsep}{5pt}			
\begin{tabular}{cccccccccccc}
	\toprule 
	Entry &  &  & \multicolumn{4}{c}{$c=0.06$} & \multicolumn{1}{c}{} & \multicolumn{4}{c}{$c=0$}\tabularnewline
	\cmidrule{4-7} \cmidrule{5-7} \cmidrule{6-7} \cmidrule{7-7} \cmidrule{9-12} \cmidrule{10-12} \cmidrule{11-12} \cmidrule{12-12} 
	age & $\alpha$ & & static & mixed & dynamic & full dyn &  & static & mixed & dynamic & full dyn\tabularnewline
	\midrule
	{\small{}$60$} & $54.80$ &  & $108.11$ & $109.20$ & $111.70$ & $111.84$ &  & $100.00$ & $101.94$ & $102.41$ & $102.56$\tabularnewline
	{\small{}$65$} & $55.36$ &  & $107.56$ & $108.22$ & $110.42$ & $110.52$ &  & $100.00$ & $101.27$ & $101.60$ & $101.71$\tabularnewline
	{\small{}$70$} & $49.13$ &  & $106.97$ & $107.32$ & $109.20$ & $109.26$ &  & $100.00$ & $100.75$ & $100.97$ & $101.04$\tabularnewline
	{\small{}$75$} & $38.24$ &  & $106.31$ & $106.47$ & $108.04$ & $108.07$ &  & $100.00$ & $100.40$ & $100.53$ & $100.57$\tabularnewline
	{\small{}$80$} & $25.04$ &  & $105.57$ & $105.63$ & $106.87$ & $106.89$ &  & $100.00$ & $100.18$ & $100.25$ & $100.27$\tabularnewline
	\bottomrule
\end{tabular}
			
	}\end{center}
	
	\caption{\label{tab:BS4}The contract price  at time $t=0$ in the Black-Scholes model, when
		$\alpha$ is set as the fair value for $c=0$ and for the static
		withdrawal strategy.}
\end{table}

  \FloatBarrier
 \subsection{The Black-Scholes CIR model}
We enrich our discussion, by assuming that the mutual fund evolves according to the Black-Scholes CIR model, namely by adding to the previous modelling setting a stochastic representation of the underlying short interest rate. In this respect, it is not possible to use the Monte Carlo control variates technique in \cite{Hsieh2018}, since there are no closed formulas for the expected values in \eqref{eq:CV} (with the only exception of $\mathbb{E}^{\mathbb{Q}}\left[\tau \right]$). Therefore, we test the performance of the Tree-LTC only against the standard Monte Carlo method. 
 
 Also in this model, we begin by testing the convergence of the Tree-LTC algorithm by comparing it with a standard Monte Carlo method. Again, we consider four numerical configurations, that are defined by the numerical setups  provided in Table \ref{tab:CIR0}, with the same arrangement used for the BS model: as far as the MC method is considered, we report the number of simulations and then the number of simulation time steps per year. As far as the Tree-LTC method is considered, we report the  number of time steps and the factor $f_A$. \myrb{In this model, given that the valuation of the fair value of $\alpha$ requires a greater computational effort, in order to speed up the contract valuation procedure, for both numerical methods,  we first perform a rough estimation of fair $\alpha$ using configuration A, and then run the procedure around this approximation using the reference configuration (B, C or D).} The fair values of $\alpha$, computed with respect to the four setups, are reported in Table \ref{tab:CIR1}. The results obtained here have similar characteristics to those obtained in the Black-Scholes model: both models produce compatible results, but the Tree-LTC method produces more stable results in less computational time.  
Table \ref{tab:CIR2} presents the outcomes as the PH's age varies, and whether LTC is included or not. The two numerical methods deliver very similar results about the age pattern of the fair value of $\alpha$.
When considering the impact of the PH's starting age, it is evident that  making the withdrawal rate $g$  vary with the initial age $x_0$ of the insured results in fair values of $\alpha$ that are closely aligned across the considered ages.
Furthermore, when comparing the cases for $c=6\%$ and $c=0$, it is clear that the addition of LTC to the insurance policy does lead to a higher fair value for $\alpha$. However, in the examined scenario, this increase never surpasses $130$ basis points, a value that is generally considered acceptable. 
 
\myrb{We deepen our analysis on the fair value of $\alpha$ by assessing how it changes with several factors changing: the PH's initial age, the withdrawal strategy (static, mixed, dynamic and full dynamic), and the fund  volatility $\sigma_F$. We report the results in Table \ref{tab:CIR4}. We can observe that the fair value of $\alpha$ increases when considering a withdrawal strategy with a wider range of withdrawal possibilities. Moreover, the inclusion of advanced withdrawal strategies does not  penalize the computational cost of the algorithm: when switching from the static strategy to the full dynamic strategy, in most cases, the computational times do not significantly increase. We also observe that the fair value of $\alpha$ increases as  $\sigma_F$ increases. The increase does not depend much on the initial age of the insured. Instead, it is more sensitive to the withdrawal strategy:  the greater the volatility, the greater the opportunities to exploit a flexible strategy efficiently. This can have implications   on the attractiveness of the insurance products due to higher contract value and thus higher contract fees. For this reason, the use of volatility risk mitigation techniques to limit $\sigma_F$ (see, for example, \cite{Berardi2024}) can be an effective choice.}
 
\myrb{As a further step in our analysis, we assess also how the fair value of $\alpha$ changes with the volatility of the interest rate, namely $\sigma_r$. We report the outcomes in Table \ref{tab:CIR3}, where we show four possible values for $\sigma_r$.   We can notice that the fair values of $\alpha$ for $\sigma_r=0.001$ are very close to those for the BS model, reported in Table \ref{tab:BS2}, since for $\sigma_r=0$ the BS-CIR model reduces to the nested BS model.
Moreover, for this particular case, the computational times are higher than usual: this is due to the fact that the smaller $\sigma_r$ is, the greater the number of nodes of the interest rate tree between $r_0$ and zero. However, in line with standard practice, we report this very particular parameter setting only for comparison purposes between the BS-CIR and the BS models, since it is not interesting from a practical point of view when the stochastic interest rate is considered.
Secondly, we observe that, as $\sigma_r$ increases, the fair value of $\alpha$ initially tends to decrease slightly, and then to increase. This particular dynamics is related to the specific choice of a negative correlation coefficient $\rho$, that has been set equal to $-0.25$. Such a choice  is  consistent with market observations, as previously explained in the beginning of Section \ref{NR}. 
In this respect, we study the effect of correlation coefficient $\rho$ on the fair value of $\alpha$. Figure \ref{fig:2} represents the fair value of $\alpha$ as a function of $\sigma_r$ for different values of the correlation parameter $\rho$, when a static withdrawal strategy is employed.  We emphasize that the values used to generate this Figure was calculated using the Tree-LTC algorithm, but, as a robustness check, they were also validated by the Monte Carlo method. As it can be seen, when the correlation parameter  is negative, the fair value of $\alpha$ (and thus the price of the contract) initially tends to decrease and then grows, whereas when $\rho$ is positive there is only growth. This aspect is important when selecting the fund to which to link the policy, preferring funds that are negatively correlated with interest rate trends.
 }
  
As a further analysis, we assess what the optimal withdrawal strategy should be under different assumptions about the amount of the account value, the interest rate, the health status of the PH and the time of contract evaluation. The results are presented in the Figure \ref{fig:1}. In this Figure, we display the dynamic optimal withdrawal strategy for a GLWB-LTC contract, as a function of the account value
$A^{2+}_n$ (x-axis) and of the interest rate $r_n$ (y-axis), varying the anniversary $n$ and the health status of the PH, $M_n$. The green region denotes the points for which it is convenient not to withdraw ($\gamma_{n} = 0$), the white region for which it is convenient to withdraw at the guaranteed minimum rate ($\gamma_{n} = 1$), and the orange region for which it is convenient to terminate the contract ($\gamma_{n} = 2$).
 Looking at the various cases analysed, we can see that when the account value takes high values, surrendering is the best choice: the cost of fees is not worth the insurance coverage provided by the contract. When analysing the effect of interest rates, we notice that the higher the interest rate, the more convenient the choice of surrender option. On the other hand, when the interest rate is low and the account value takes on values close to the initial premium, the most convenient choice is not to withdraw and thus to let the benefit base increase in its value. This implies to reserve a higher LTC payment in case of disability at the subsequent anniversaries and/or a higher withdrawal. This aspect emphasizes an important advantage of flexible withdrawal strategies, especially in relation to insurance policies offering protection from health risks. Indeed, the PH is given more choice about how much to save for protection from the possible disability states at future times. When analyzing the optimal withdrawal strategy with respect to time, we notice that as the PH grows old, the most convenient is for her to withdraw at the minimum guaranteed rate and thus to undertake a decumulation strategy as the component of disability protection becomes more and more important. The passing of time has a further effect: in the early years it is less convenient to terminate the contract early because of the  cost charged for withdrawals beyond the guaranteed minimum. After seven years this penalty disappears and it is then more convenient to surrender. Finally, when the PH is very old, the most convenient choice is a standard withdrawal, almost always: it is neither worthwhile to give up a withdrawal (the cost of giving up does not pay off over time) nor to terminate the contract (the PH loses the insurance coverage).
 \begin{table}[h] 
 	\begin{center}
 		\scalebox{1}{  
 			\begin{tabular}{cccc}
 				\toprule
 				Setup &  & MC & Tree-LTC\tabularnewline
 				\midrule
 				{\small{}$A$} &  & $1\cdot10^{6},\phantom{1}25$ & $\phantom{1}25,100$\tabularnewline
 				{\small{}$B$} &  & $2\cdot10^{6},\phantom{1}50$ & $\phantom{1}50,200$\tabularnewline
 				{\small{}$C$} &  & $4\cdot10^{6},100$ & $100,400$\tabularnewline
 				{\small{}$D$} &  & $8\cdot10^{6},200$ & $200,800$\tabularnewline
 				\bottomrule
 			\end{tabular}
 		}
 			\caption{\label{tab:CIR0}Parameter setup for the numerical algorithms when the BS-CIR model is considered.}
  \end{center}
\end{table}
   
  \begin{table}[h] 
 	\begin{center} 
 	\scalebox{1}{ 
 		
\begin{tabular}{ccccccc}
	\toprule 
	&  & \multicolumn{2}{c}{$c=0.06$} &  & \multicolumn{2}{c}{$c=0.00$}\tabularnewline
	\cmidrule{3-7} \cmidrule{4-7} \cmidrule{5-7} \cmidrule{6-7} \cmidrule{7-7} 
	Setup &  & MC & Tree-LTC &  & MC & Tree-LTC\tabularnewline
	\midrule
	{\small{}$A$} &  & $\underset{\left(32\right)}{157.16\pm3.56}$ & $\underset{\left(3\right)}{159.24}$ &  & $\underset{\left(32\right)}{53.47\pm3.05}$ & $\underset{\left(4\right)}{54.61}$\tabularnewline
	{\small{}$B$} &  & $\underset{\left(151\right)}{159.14\pm2.63}$ & $\underset{\left(14\right)}{159.39}$ &  & $\underset{\left(150\right)}{53.58\pm2.24}$ & $\underset{\left(15\right)}{54.92}$\tabularnewline
	{\small{}$C$} &  & $\underset{\left(699\right)}{158.76\pm1.85}$ & $\underset{\left(81\right)}{159.44}$ &  & $\underset{\left(500\right)}{55.03\pm1.54}$ & $\underset{\left(75\right)}{55.00}$\tabularnewline
	{\small{}$D$} &  & $\underset{\left(1804\right)}{159.84\pm1.30}$ & $\underset{\left(431\right)}{159.45}$ &  & $\underset{\left(1508\right)}{55.17\pm1.04}$ & $\underset{\left(379\right)}{55.02}$\tabularnewline
	\bottomrule
\end{tabular} }\end{center}
 	\caption{\label{tab:CIR1}  The fair values of $\alpha$ (in basis points), in the BS-CIR model, by changing the numerical setup and by assuming  the presence ($c=6\%$) or absence ($c=0$) of LTC. The values in parentheses indicate computational time in seconds.}
\end{table}
\begin{table}[h]
	\begin{center}
		\scalebox{1}{ %
\begin{tabular}{ccccccc}
	\toprule 
	Entry &  & \multicolumn{2}{c}{$c=0.06$} &  & \multicolumn{2}{c}{$c=0.00$}\tabularnewline
	\cmidrule{3-7} \cmidrule{4-7} \cmidrule{5-7} \cmidrule{6-7} \cmidrule{7-7} 
	age &  & MC & Tree-LTC &  & MC & Tree-LTC\tabularnewline
	\midrule
	{\small{}$60$} &  & $\underset{\left(1804\right)}{159.84\pm1.30}$ & $\underset{\left(81\right)}{159.64}$ &  & $\underset{\left(1885\right)}{55.81\pm1.09}$ & $\underset{\left(75\right)}{55.00}$\tabularnewline
	{\small{}$65$} &  & $\underset{\left(1736\right)}{170.10\pm1.40}$ & $\underset{\left(60\right)}{170.98}$ &  & $\underset{\left(1706\right)}{54.49\pm1.12}$ & $\underset{\left(61\right)}{55.25}$\tabularnewline
	{\small{}$70$} &  & $\underset{\left(2251\right)}{169.89\pm1.45}$ & $\underset{\left(55\right)}{169.63}$ &  & $\underset{\left(1581\right)}{48.62\pm1.17}$ & $\underset{\left(55\right)}{48.64}$\tabularnewline
	{\small{}$75$} &  & $\underset{\left(2042\right)}{159.14\pm1.49}$ & $\underset{\left(50\right)}{158.45}$ &  & $\underset{\left(1478\right)}{38.06\pm1.21}$ & $\underset{\left(52\right)}{37.43}$\tabularnewline
	{\small{}$80$} &  & $\underset{\left(1285\right)}{140.32\pm1.51}$ & $\underset{\left(45\right)}{140.70}$ &  & $\underset{\left(1277\right)}{23.97\pm1.26}$ & $\underset{\left(44\right)}{24.06}$\tabularnewline
	\bottomrule
\end{tabular}
			
	}\end{center}
	
	\caption{ \label{tab:CIR2}   The fair values of $\alpha$ (in basis points),  in the BS-CIR model, by changing the entry age of the PH and by assuming  the presence ($c=6\%$)
		or absence ($c=0$) of LTC. The values in parentheses indicate computational
		time in seconds.}
\end{table}

\begin{table}[h]
	\begin{center}
		\scalebox{0.8}{ 
			\setlength{\tabcolsep}{3pt} 
			\begin{tabular}{cccccccccccccccc}
				\toprule 
				Entry &  & \multicolumn{4}{c}{BS CIR-with $\sigma_{F}=0.15$} &  & \multicolumn{4}{c}{BS CIR-with $\sigma_{F}=0.20$} &  & \multicolumn{4}{c}{BS CIR-with $\sigma_{F}=0.25$}\tabularnewline
				age &  & static & mixed & dynamic & full dyn &  & static & mixed & dynamic & full dyn &  & static & mixed & dynamic & full dyn\tabularnewline
				\midrule
				{\small{}$60$} &  & $\underset{\left(89\right)}{121.61}$ & $\underset{\left(110\right)}{150.88}$ & $\underset{\left(134\right)}{162.81}$ & $\underset{\left(143\right)}{168.43}$ &  & $\underset{\left(81\right)}{159.64}$ & $\underset{\left(80\right)}{220.69}$ & $\underset{\left(111\right)}{234.38}$ & $\underset{\left(77\right)}{249.25}$ &  & $\underset{\left(51\right)}{198.98}$ & $\underset{\left(72\right)}{305.88}$ & $\underset{\left(86\right)}{321.49}$ & $\underset{\left(84\right)}{353.65}$\tabularnewline
				{\small{}$65$} &  & $\underset{\left(80\right)}{133.88}$ & $\underset{\left(104\right)}{153.75}$ & $\underset{\left(125\right)}{165.03}$ & $\underset{\left(132\right)}{169.18}$ &  & $\underset{\left(60\right)}{170.98}$ & $\underset{\left(394\right)}{213.79}$ & $\underset{\left(96\right)}{226.45}$ & $\underset{\left(112\right)}{237.26}$ &  & $\underset{\left(47\right)}{210.45}$ & $\underset{\left(66\right)}{286.56}$ & $\underset{\left(80\right)}{300.81}$ & $\underset{\left(85\right)}{324.15}$\tabularnewline
				{\small{}$70$} &  & $\underset{\left(74\right)}{134.09}$ & $\underset{\left(95\right)}{146.09}$ & $\underset{\left(114\right)}{156.19}$ & $\underset{\left(120\right)}{158.79}$ &  & $\underset{\left(55\right)}{169.63}$ & $\underset{\left(68\right)}{196.81}$ & $\underset{\left(87\right)}{208.03}$ & $\underset{\left(101\right)}{214.90}$ &  & $\underset{\left(43\right)}{208.09}$ & $\underset{\left(60\right)}{258.22}$ & $\underset{\left(75\right)}{270.83}$ & $\underset{\left(73\right)}{286.12}$\tabularnewline
				{\small{}$75$} &  & $\underset{\left(66\right)}{124.94}$ & $\underset{\left(81\right)}{131.39}$ & $\underset{\left(104\right)}{140.03}$ & $\underset{\left(107\right)}{141.42}$ &  & $\underset{\left(50\right)}{158.45}$ & $\underset{\left(71\right)}{173.91}$ & $\underset{\left(85\right)}{183.48}$ & $\underset{\left(91\right)}{187.21}$ &  & $\underset{\left(39\right)}{195.35}$ & $\underset{\left(54\right)}{225.34}$ & $\underset{\left(63\right)}{236.13}$ & $\underset{\left(64\right)}{244.91}$\tabularnewline
				{\small{}$80$} &  & $\underset{\left(58\right)}{109.63}$ & $\underset{\left(73\right)}{112.68}$ & $\underset{\left(91\right)}{119.81}$ & $\underset{\left(101\right)}{120.48}$ &  & $\underset{\left(45\right)}{140.70}$ & $\underset{\left(58\right)}{148.31}$ & $\underset{\left(77\right)}{156.18}$ & $\underset{\left(87\right)}{157.83}$ &  & $\underset{\left(35\right)}{175.50}$ & $\underset{\left(47\right)}{191.08}$ & $\underset{\left(52\right)}{199.98}$ & $\underset{\left(57\right)}{204.06}$\tabularnewline
				\bottomrule
			\end{tabular}
	}\end{center}
	
	\caption{\label{tab:CIR4}  The fair values of $\alpha$ (in basis points), in the BS-CIR model, for a GLWB-LTC with $c=6\%$,  by changing
		the entry age of the PH, the withdrawal strategy and the volatility $\sigma_F$ of the fund. The values in parentheses indicate
		computational time in seconds.}
\end{table}

\begin{table}[h]
	\begin{center}
		\scalebox{0.9}{  
		\begin{tabular}{ccccccccccc}
			\toprule 
			Entry &  & \multicolumn{4}{c}{BS CIR-with $\sigma_{r}=0.001$} &  & \multicolumn{4}{c}{BS CIR-with $\sigma_{r}=0.05$}\tabularnewline
			age &  & static & mixed & dynamic & full dyn &  & static & mixed & dynamic & full dyn\tabularnewline
			\midrule
			{\small{}$60$} &  & $\underset{\left(801\right)}{154.35}$ & $\underset{\left(1044\right)}{216.78}$ & $\underset{\left(1308\right)}{229.35}$ & $\underset{\left(993\right)}{244.12}$ &  & $\underset{\left(66\right)}{153.72}$ & $\underset{\left(84\right)}{213.91}$ & $\underset{\left(85\right)}{226.70}$ & $\underset{\left(86\right)}{240.84}$\tabularnewline
			{\small{}$65$} &  & $\underset{\left(801\right)}{166.75}$ & $\underset{\left(959\right)}{210.99}$ & $\underset{\left(1201\right)}{222.64}$ & $\underset{\left(1265\right)}{233.52}$ &  & $\underset{\left(60\right)}{165.97}$ & $\underset{\left(74\right)}{208.45}$ & $\underset{\left(88\right)}{220.30}$ & $\underset{\left(98\right)}{230.71}$\tabularnewline
			{\small{}$70$} &  & $\underset{\left(814\right)}{166.70}$ & $\underset{\left(885\right)}{195.14}$ & $\underset{\left(1056\right)}{205.49}$ & $\underset{\left(1094\right)}{212.50}$ &  & $\underset{\left(60\right)}{165.69}$ & $\underset{\left(68\right)}{192.87}$ & $\underset{\left(82\right)}{203.42}$ & $\underset{\left(96\right)}{210.10}$\tabularnewline
			{\small{}$75$} &  & $\underset{\left(761\right)}{156.84}$ & $\underset{\left(793\right)}{173.20}$ & $\underset{\left(953\right)}{182.06}$ & $\underset{\left(996\right)}{185.91}$ &  & $\underset{\left(52\right)}{155.62}$ & $\underset{\left(64\right)}{171.19}$ & $\underset{\left(74\right)}{180.22}$ & $\underset{\left(86\right)}{183.87}$\tabularnewline
			{\small{}$80$} &  & $\underset{\left(577\right)}{140.19}$ & $\underset{\left(715\right)}{148.34}$ & $\underset{\left(859\right)}{155.63}$ & $\underset{\left(885\right)}{157.37}$ &  & $\underset{\left(45\right)}{138.84}$ & $\underset{\left(54\right)}{146.56}$ & $\underset{\left(65\right)}{154.00}$ & $\underset{\left(115\right)}{155.65}$\tabularnewline
			\midrule
			&  &  &  &  &  &  &  &  &  & \tabularnewline
			&  & \multicolumn{4}{c}{BS CIR-with $\sigma_{r}=0.10$} &  & \multicolumn{4}{c}{BS-CIR with $\sigma_{r}=0.15$}\tabularnewline
			&  & static & mixed & dynamic & full dyn &  & static & mixed & dynamic & full dyn\tabularnewline
			\midrule
			{\small{}$60$} &  & $\underset{\left(81\right)}{159.64}$ & $\underset{\left(80\right)}{220.69}$ & $\underset{\left(111\right)}{234.38}$ & $\underset{\left(77\right)}{249.25}$ &  & $\underset{\left(63\right)}{171.59}$ & $\underset{\left(105\right)}{237.22}$ & $\underset{\left(110\right)}{252.52}$ & $\underset{\left(113\right)}{269.54}$\tabularnewline
			{\small{}$65$} &  & $\underset{\left(60\right)}{170.98}$ & $\underset{\left(394\right)}{213.79}$ & $\underset{\left(96\right)}{226.45}$ & $\underset{\left(112\right)}{237.26}$ &  & $\underset{\left(58\right)}{181.77}$ & $\underset{\left(88\right)}{226.99}$ & $\underset{\left(107\right)}{241.01}$ & $\underset{\left(103\right)}{253.12}$\tabularnewline
			{\small{}$70$} &  & $\underset{\left(55\right)}{169.63}$ & $\underset{\left(68\right)}{196.81}$ & $\underset{\left(87\right)}{208.03}$ & $\underset{\left(101\right)}{214.90}$ &  & $\underset{\left(57\right)}{178.42}$ & $\underset{\left(76\right)}{206.84}$ & $\underset{\left(96\right)}{219.18}$ & $\underset{\left(94\right)}{226.76}$\tabularnewline
			{\small{}$75$} &  & $\underset{\left(50\right)}{158.45}$ & $\underset{\left(71\right)}{173.91}$ & $\underset{\left(85\right)}{183.48}$ & $\underset{\left(91\right)}{187.21}$ &  & $\underset{\left(52\right)}{165.21}$ & $\underset{\left(52\right)}{181.25}$ & $\underset{\left(83\right)}{191.70}$ & $\underset{\left(83\right)}{195.77}$\tabularnewline
			{\small{}$80$} &  & $\underset{\left(45\right)}{140.70}$ & $\underset{\left(58\right)}{148.31}$ & $\underset{\left(77\right)}{156.18}$ & $\underset{\left(87\right)}{157.83}$ &  & $\underset{\left(46\right)}{145.67}$ & $\underset{\left(61\right)}{153.50}$ & $\underset{\left(73\right)}{162.04}$ & $\underset{\left(75\right)}{163.81}$\tabularnewline
			\bottomrule
		\end{tabular}
			
	}\end{center}
	
		\caption{\label{tab:CIR3}  The fair values of $\alpha$ (in basis points), in the BS-CIR model, for a GLWB-LTC with $c=6\%$,  by changing
		the entry age of the PH, the withdrawal strategy and the volatility $\sigma_r$ of the interest rate. The values in parentheses indicate
		computational time in seconds.}
\end{table}

\begin{figure}[h]
	\begin{center}
		\includegraphics[width=0.7\textwidth]{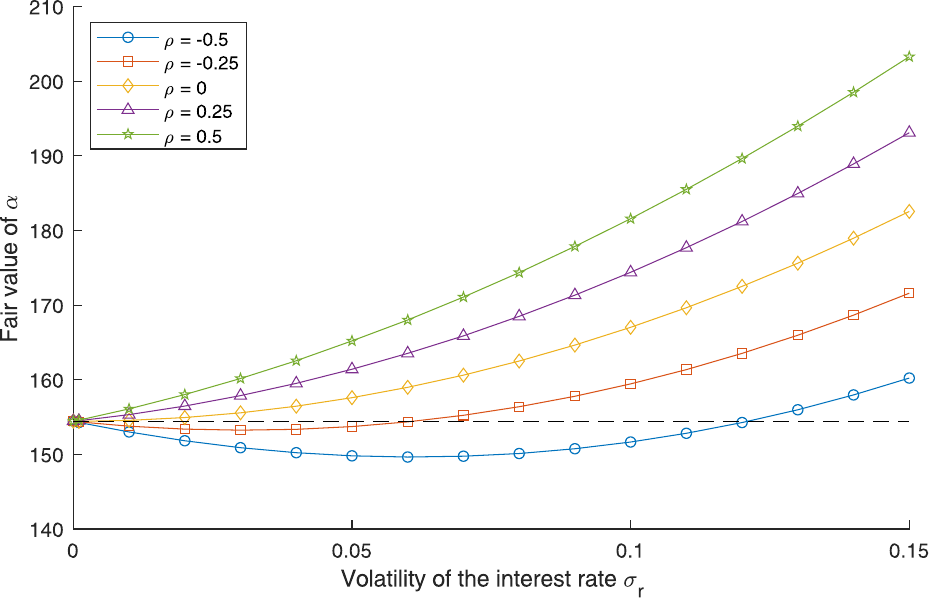}
		\end{center}
		\caption{\label{fig:2}The fair value  of $\alpha$ (in basis points), in the BS-CIR model, for a GLWB-LTC with $c=6\%$,  $X_0=60$, and different value of $\rho$. The withdrawal strategy is assumed to be the static one. The black dotted line has a constant y-value equal to the fair value of alpha in the Black-Scholes model.}
\end{figure}

\begin{figure}[h]
	\begin{center}
		\includegraphics[width=\textwidth]{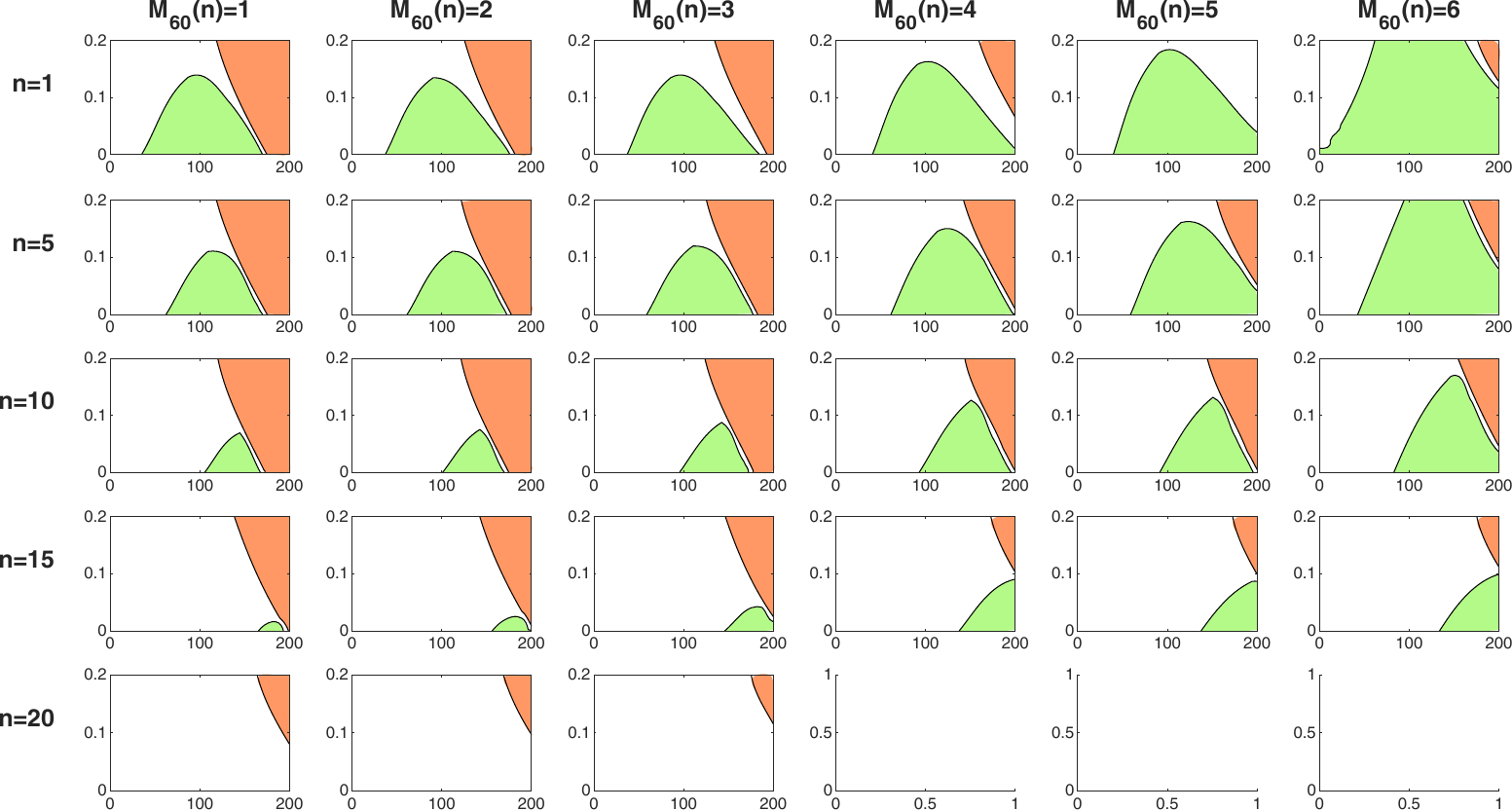}
	\end{center}
	\caption{\label{fig:1}Dynamic optimal withdrawal strategy for a GLWB-LTC contract, as a function of account value $A^{2+}_n$ (x-axis) and interest rate $r_n$ (y-axis), varying the anniversary $n$ and the health status of the insured $M_n$. The green region denotes the points for which it is convenient not to withdraw ($\gamma_{n}=0$), the white region for which it is convenient to withdraw at the guaranteed minimum rate ($\gamma_{n}=1$), and the orange region for which it is convenient to terminate the contract ($\gamma_{n}=2$).}
\end{figure}

 \FloatBarrier
	\section{Conclusions}\label{sec:conclusion}  
	Ageing and disability cannot be disentangled.
	Private insurance sector can fill the gaps of public social protection systems.
	It is important to overcome the small market penetration of LTC private insurance, by making insurance products more attractive to the demand-side.

	The state-of-the-art has proposed insurance products that bundle longevity, disability and downside risks. For instance, \cite{Hsieh2018} proposed the LCA-GLWB insurance contract, namely variable long-term care annuities, granting the policyholder to withdraw a contractually defined fraction of the benefit base until she remains alive. The key idea of our paper is to provide more general features for this insurance product and to refine its pricing method. In particular, we depart from the existing literature on variable long-term care annuities by introducing the opportunity for the policyholder to choose how much to withdraw (dynamic withdrawal strategy), including the surrender option. We name ``GLWB-LTC" the insurance product embedding LTC payouts and dynamic withdrawals. 
	
	The state-of-the-art emphasizes, in relation to GLWB variable annuity contracts, that the surrender option is generally attractive to the demand side, as policyholders may be less prone to perceive insurance securities as illiquid investments. Through our numerical results, coming from the pricing of the GLWB-LTC product, we show that the dynamic withdrawal strategy acquires even more relevance within the GLWB annuity product that offers protection from disability risks.

Our numerical analysis leads to interesting and original findings with important implications on both the social value and the attractivity of the GLWB-LTC product. These findings mainly relate to the advantages of the dynamic withdrawal strategy, also under the policyholder's perspective.  In particular, the policyholder can choose more flexibly how much to save for protection from the possible state of disabled at future times. 

\myrb{Future research paths could explore how the fund volatility impacts on the valuation of the GLWB-LTC product and methods to refine its design for better protection against this risk. This could be important since this kind of products combines insurance and investment components.}

	\section*{Conflict of interest}    
	The authors have no conflicts of interest to declare.
	\FloatBarrier 
	\bibliography{ref}
	\appendix
	
	\section{Appendix: transition intensities} \label{appendix:ti}
	Following \cite{Pritchard2006}, the intensities for the health state transitions are computed from the coefficient reported in Table \ref{TA1}. Specifically, depending on which provides
	a better fit, for $i\in \left\lbrace 1,\dots, 6 \right\rbrace$ and $j\in \left\lbrace 1,\dots, 6, 7 \right\rbrace \setminus \left\lbrace i\right\rbrace $, the transition intensities are defined as
	\begin{equation*}
		q^{x_0}_{i, j}= A_{i,j}+B_{i,j}\cdot \exp \left(C_{i,j}\left(x_0+t-68.5 \right)  \right) 
	\end{equation*}
	or 
	\begin{equation*}
		q^{x_0}_{i, j}= A_{i,j}+D_{i,j}\cdot  \left(x_0+t\right)  
	\end{equation*}
	with a lower bound of zero on all intensities at all ages. In addition for $i\in \left\lbrace 1,\dots, 6 \right\rbrace$
	\begin{equation*}
		q^{x_0}_{i, i}= -\sum_{j\neq i} q^{x_0}_{i, j}.
	\end{equation*}
	Finally, for  $j\in \left\lbrace 1,\dots, 7 \right\rbrace, q^{x_0}_{7, j}=0$.

	\begin{table}[h]
		\begin{centering}
			\begin{tabular}{lllcccc}
				\toprule 
				&  &  & \multicolumn{4}{c}{Parameters}\tabularnewline
				\cmidrule{4-7} \cmidrule{5-7} \cmidrule{6-7} \cmidrule{7-7} 
				From state & To state &  & A & B & C & D\tabularnewline
				\midrule 
				& IADL Only &  & $-3.22\cdot10^{-2}$ & $5.19\cdot10^{-2}$ & $\phantom{-}4.35\cdot10^{-2}$ & $-$\tabularnewline
				& 1-2 ADLs &  & $\phantom{-}9.58\cdot10^{-3}$ & $2.11\cdot10^{-3}$ & $\phantom{-}1.74\cdot10^{-1}$ & $-$\tabularnewline
				\multirow{2}{*}{Healty} & 3-4 ADLs &  & $-2.34\cdot10^{-2}$ & $-$ & $-$ & $\phantom{-}3.85\cdot10^{-4}$\tabularnewline
				& 5-6 ADLs &  & $-1.37\cdot10^{-4}$ & $3.16\cdot10^{-3}$ & $\phantom{-}8.01\cdot10^{-2}$ & $-$\tabularnewline
				& Inst'd &  & $-9.05\cdot10^{-4}$ & $3.15\cdot10^{-3}$ & $\phantom{-}1.32\cdot10^{-1}$ & $-$\tabularnewline
				& Dead &  & $-1.62\cdot10^{-1}$ & $-$ & $-$ & $\phantom{-}2.64\cdot10^{-3}$\tabularnewline
				\midrule 
				& Healty &  & $\phantom{-}1.04\cdot10^{-0}$ & $-$ & $-$ & $-1.13\cdot10^{-2}$\tabularnewline
				& 1-2 ADLs &  & $-3.38\cdot10^{-1}$ & $-$ & $-$ & $\phantom{-}8.32\cdot10^{-3}$\tabularnewline
				IADL & 3-4 ADLs &  & $\phantom{-}2.94\cdot10^{-2}$ & $-$ & $-$ & $-1.59\cdot10^{-4}$\tabularnewline
				Only & 5-6 ADLs &  & $-9.89\cdot10^{-2}$ & $1.33\cdot10^{-1}$ & $\phantom{-}8.16\cdot10^{-3}$ & $-$\tabularnewline
				& Inst'd &  & $-1.81\cdot10^{-1}$ & $-$ & $-$ & $\phantom{-}2.90\cdot10^{-3}$\tabularnewline
				& Dead &  & $-3.19\cdot10^{-2}$ & $8.80\cdot10^{-2}$ & $\phantom{-}1.60\cdot10^{-2}$ & $-$\tabularnewline
				\midrule 
				& Healty &  & $\phantom{-}1.74\cdot10^{-1}$ & $-$ & $-$ & $-1.45\cdot10^{-3}$\tabularnewline
				& IADL Only &  & $\phantom{-}5.45\cdot10^{-1}$ & $-$ & $-$ & $-4.71\cdot10^{-3}$\tabularnewline
				1-2 & 3-4 ADLs &  & $\phantom{-}1.85\cdot10^{-1}$ & $5.62\cdot10^{-3}$ & $\phantom{-}1.33\cdot10^{-1}$ & $-$\tabularnewline
				ADLs & 5-6 ADLs &  & $-6.10\cdot10^{-2}$ & $1.04\cdot10^{-1}$ & $-1.11\cdot10^{-2}$ & $-$\tabularnewline
				& Inst'd &  & $-5.61\cdot10^{-2}$ & $7.72\cdot10^{-2}$ & $\phantom{-}3.48\cdot10^{-2}$ & $-$\tabularnewline
				& Dead &  & $-4.68\cdot10^{-2}$ & $-$ & $-$ & $\phantom{-}1.93\cdot10^{-3}$\tabularnewline
				\midrule 
				& Healty &  & $\phantom{-}1.03\cdot10^{-1}$ & $-$ & $-$ & $-1.11\cdot10^{-3}$\tabularnewline
				& IADL Only &  & $-4.26\cdot10^{-3}$ & $2.14\cdot10^{-3}$ & $\phantom{-}1.48\cdot10^{-1}$ & $-$\tabularnewline
				3-4 & 1-2ADLs &  & $\phantom{-}1.61\cdot10^{-0}$ & $-$ & $-$ & $-1.69\cdot10^{-2}$\tabularnewline
				ADLs & 5-6 ADLs &  & $\phantom{-}1.64\cdot10^{-2}$ & $2.13\cdot10^{-1}$ & $\phantom{-}4.51\cdot10^{-2}$ & $-$\tabularnewline
				& Inst'd &  & $-9.20\cdot10^{-2}$ & $1.09\cdot10^{-1}$ & $\phantom{-}3.52\cdot10^{-2}$ & $-$\tabularnewline
				& Dead &  & $\phantom{-}1.27\cdot10^{-1}$ & $-$ & $-$ & $-5.50\cdot10^{-4}$\tabularnewline
				\midrule 
				& Healty &  & $\phantom{-}1.06\cdot10^{-1}$ & $-$ & $-$ & $-9.93\cdot10^{-4}$\tabularnewline
				& IADL Only &  & $\phantom{-}2.85\cdot10^{-1}$ & $-$ & $-$ & $-3.08\cdot10^{-3}$\tabularnewline
				5-6 & 1-2ADLs &  & $-1.81\cdot10^{-1}$ & $2.23\cdot10^{-1}$ & $\phantom{-}4.62\cdot10^{-3}$ & $-$\tabularnewline
				ADLs & 3-4 ADLs &  & $\phantom{-}1.40\cdot10^{-1}$ & $-$ & $-$ & $\phantom{-}3.16\cdot10^{-4}$\tabularnewline
				& Inst'd &  & $-2.00\cdot10^{-1}$ & $-$ & $-$ & $\phantom{-}3.80\cdot10^{-3}$\tabularnewline
				& Dead &  & $\phantom{-}1.76\cdot10^{-1}$ & $4.53\cdot10^{-2}$ & $\phantom{-}5.28\cdot10^{-2}$ & $-$\tabularnewline
				\midrule 
				& Healty &  & $\phantom{-}2.39\cdot10^{-3}$ & $2.84\cdot10^{-2}$ & $-1.19\cdot10^{-1}$ & $-$\tabularnewline
				& IADL Only &  & $\phantom{-}2.89\cdot10^{-2}$ & $-$ & $-$ & $-2.90\cdot10^{-4}$\tabularnewline
				\multirow{2}{*}{Inst'd} & 1-2ADLs &  & $-3.10\cdot10^{-2}$ & $3.89\cdot10^{-2}$ & $-1.02\cdot10^{-2}$ & $-$\tabularnewline
				& 3-4 ADLs &  & $-1.94\cdot10^{-1}$ & $2.05\cdot10^{-1}$ & $-3.68\cdot10^{-4}$ & $-$\tabularnewline
				& 5-6 ADLs &  & $\phantom{-}9.87\cdot10^{-3}$ & $-$ & $-$ & $-6.85\cdot10^{-5}$\tabularnewline
				& Dead &  & $-5.71\cdot10^{-1}$ & $-$ & $-$ & $\phantom{-}9.98\cdot10^{-3}$\tabularnewline
				\bottomrule
			\end{tabular}
			\par\end{centering}
		\caption{\label{TA1}Parameters reported in \cite{Pritchard2006} for computing transition intensities.}
		
	\end{table}
	
	
	\FloatBarrier 
	\section{Appendix: upper bound for the interest rate tree}\label{appendix:ub}
	\myrb{In this Appendix, we demonstrate the existence of a value, denoted as $R_{\max}(i)$, such that for the interest rate tree, all nodes below $R_{\max}(i)$ exclusively have successors that are smaller than $R_{\max}(i)$. To determine this value, we begin by solving the equation presented below, with respect to $\overline{k}(i)$:		
	\begin{equation}
		\label{kmax} k_d(i,\overline{k}(i))=\overline{k}(i)-1.
	\end{equation} 
	By solving equation \eqref{kmax}, one obtains 
	$$\overline{k}(i)=
	\left\lceil 
	\frac{-2 \left(\Delta t \right) ^{3/2} k_{r}  \sqrt{r_0} \sigma_{r} +\sqrt{\Delta t^2 \sigma_{r} ^2 \left(\Delta t k_r  \left(4 \theta  k_{r} -\sigma_{r} ^2\right)+\sigma_{r} ^2\right)}+\Delta t^2 i k_{r}  \sigma_{r} ^2+\Delta t \sigma_{r} ^2}{2 \Delta t^2 k_{r}  \sigma_{r} ^2}
	\right\rceil .$$
	We define $R_{\max}(i)=\mathcal{R}_{i,\overline{k}(i)}$. It is possible to prove that
	\begin{equation} \label{ineq}
			\theta < R_{\max}(i)<\overline{R}=\frac{\left(\sqrt{\Delta t ^2 \sigma_{r} ^2 \left(\Delta t k_{r}  \left(4 \theta  k_{r} -\sigma_{r} ^2\right)+\sigma_{r} ^2\right)}+4 \Delta t^2 k_{r}  \sigma_{r} ^2+\Delta t \sigma_{r} ^2\right)^2}{4 \Delta t^3 k_{r} ^2 \sigma_{r} ^2} 
	\end{equation} 
	so that $R_{\max}(i)$ is bounded by a constant $\overline{R}$ which does not depend on $i$.
	Now,  if $k\leq\overline{k}(i)$, then 
	$$\mathcal{R}_{i+1,k_d(i,k)}<\mathcal{R}_{i+1,k_u(i,k)}\leq\mathcal{R}_{i+1,k_u(i,\overline{k}(i))} =\mathcal{R}_{i+1,\overline{k}(i)}<\mathcal{R}_{i,\overline{k}(i)},$$
	as we exploited the relation $k_u(i,\overline{k}(i))=\overline{k}(i)$ which comes from equations \eqref{10}, \eqref{kmax} and \eqref{ineq}.	
Therefore,  $\mathcal{R}_{i,k}\leq \mathcal{R}_{i,\overline{k}(i)}$ implies that both $\mathcal{R}_{i+1,k_d(i)} $ and $\mathcal{R}_{i+1,k_u(i)}$ are smaller or equal to $\mathcal{R}_{i,\overline{k}(i)} $. Thus, starting from $r_0<R_{\max}(i)$, it is not possible to reach the nodes above $R_{\max}(i)$. So, we have proved  that at time $i$ the discrete process $\bar{r}$ cannot assume value grater than $\mathcal{R}_{i,\overline{k}(i)}$. Moreover, the node $\mathcal{R}_{i,\overline{k}(i)}$ may be unreachable itself if  $k_u(i-1,\overline{k}(i-1))<\overline{k}(i)$. Therefore, one can improve even more the selection of the nodes by considering only the value of $k$ smaller than $\overline{k}(i)$ and  $k_u(i-1,\overline{k}(i-1))$. Finally, we set $k_{\max}(i)=\min \left\lbrace \overline{k}(i),i,k_u(i-1,\overline{k}(i-1))\right\rbrace $.
To conclude, at time step $i$, the only nodes to be considered are those with $k$ between $k_{\min}(i) $ and $k_{\max}(i) $, with $k_{\min}(i) $ defined in Section \ref{ss:tree_for_r}. 
Figure \ref{fig:tree}  shows an example of the tree construction. The red points are nodes $\mathcal{R}_{i,k}$ that satisfy the relation $k_{\min}(i)\leq k\leq k_{\max}(i) $ and are those actually used in the Tree-LTC algorithm, while the blue points are the first order approximation of their expected values at next time step.  The green points are the nodes that satisfy the relation $k_{\max}(i)<k\leq  \overline{k}(i)$. These green nodes, although their value is less than $R_{\max}(i)$, are unreachable and can be discarded. The grey points are the nodes that satisfy the relation $\overline{k}(i)<k\leq i $. The dotted black line corresponds to the constant $\overline{R}$ which is greater than $R_{\max}(i)$ for all values of $i$.
Again, we stress out that that all the grey and green nodes $\mathcal{R}_{i,k}$   cannot be reached from the initial node $\mathcal{R}_{0,0}$, so they can be discarded and do not need to be processed, thus reducing the computational cost of the Tree-LTC algorithm.  
\begin{figure}[h]
\begin{center}
	\includegraphics[width=\textwidth]{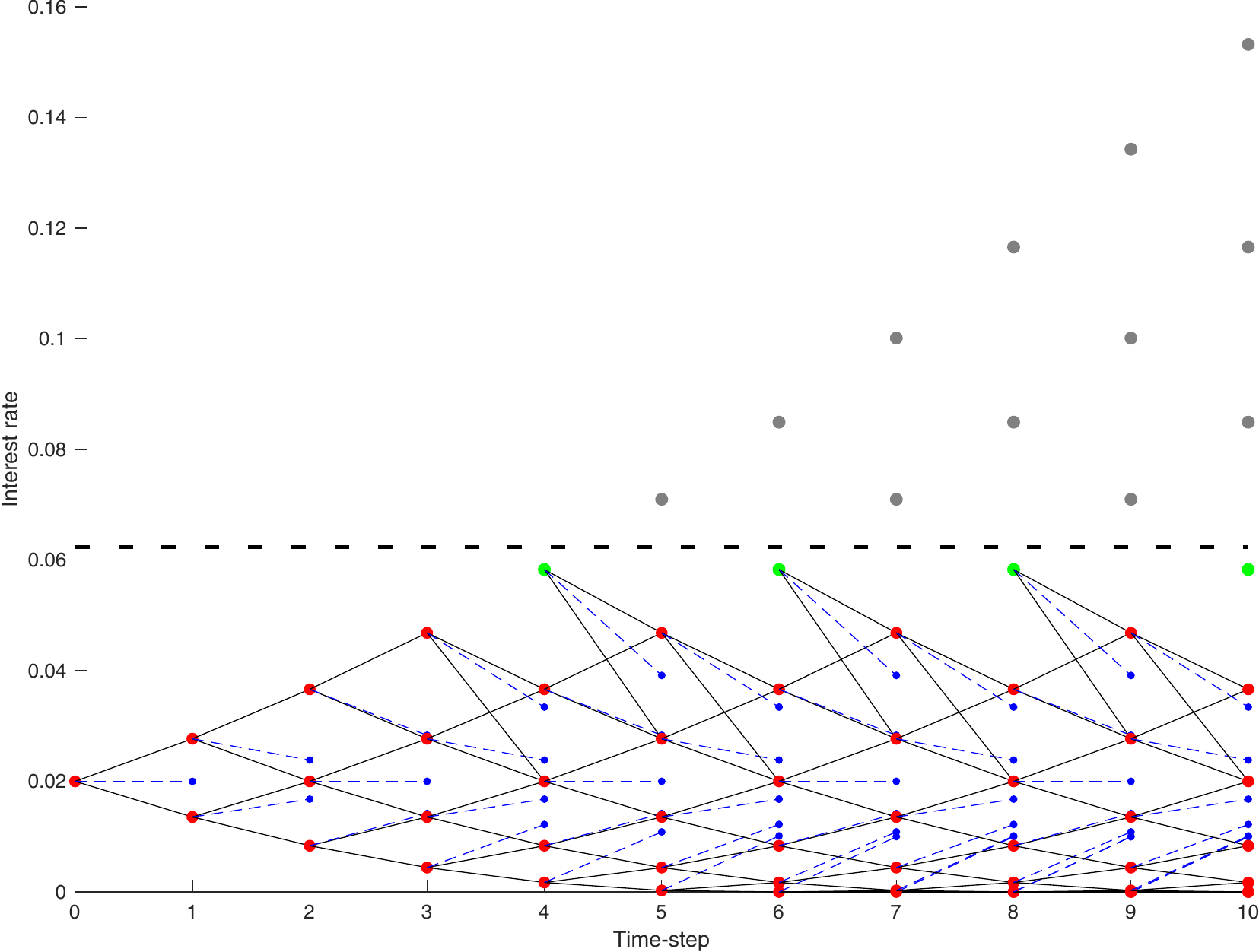}
\end{center}
\caption{\label{fig:tree}An example of construction of the interest rate tree. 
	Parameters used: $r_0=\theta=0.02, k_r=1,\sigma_r=0.05,T=5,\Delta t=0.5$.}
\end{figure}}

		\section{Appendix: analysis of the transition probabilities} \label{appendix:co}
		\myrb{In this Appendix we discuss the transition probabilities of the proposed bi-dimensional tree, employed within the Tree-LTC algorithm. In particular, our method differs from the method of \cite{Appolloni2015}   for 3 reasons: with reference to the tree for the interest rate, the nodes with zero repeated value are discarded  (we only keep one node with zero value); furthermore, when $\mathcal{R}_{i,k}\geq \theta$, the up node is defined as the down node plus 1; finally, with reference to the tree for the underlying, all nodes between two extreme values are considered and a linearity condition at the edges is exploited to estimate the value of the contract outside the considered nodes. We go through Lemma 2, Lemma 4 and Proposition 6 presented in  \cite{Appolloni2015} and we show step by step how to re-adapt the original proofs.  \\ In the following, let us consider a node $\mathcal{R}_{i,k}$ such that $k_{\min}(i) \leq k\leq  k_{\max}(i) $. 
		\setcounter{lemma}{1}
		\begin{lemma}
Let \(\theta_* < \theta\) and \(\theta^* > \theta\) be such that 
\[
0 < \theta_* < \frac{\theta \land r_0}{2} \quad \text{and} \quad \theta^* > 2(\theta \lor r_0).
\]
Then there exists a positive constant \( h_1 = h_1(\theta_*, \theta^*, k, \theta, \sigma_r)<1 \) such that for every \( \Delta t < h_1 \) the following statements hold.
\begin{enumerate}
	\item[(i)] If \( 0 \leq \mathcal{R}_{i,k} < \theta_* \sqrt{\Delta t} \) then \( k_{d}(i, k)\in\left\lbrace k,k+1 \right\rbrace  \) and \( k_{u}(i, k) \in \{k + 1, \ldots, i + 1\}\). Moreover, there exists a positive constant \( C_* > 0 \) such that
	\begin{equation}\label{lemma2_27a}
	|\mathcal{R}_{i+1,k_{d}(i,k)} - \mathcal{R}_{i,k}| \leq C_* \left( \Delta t\right) ^{3/4},
	\end{equation}
	and
	\begin{equation}\label{lemma2_27b}
	|\mathcal{R}_{i+1,k_{u}(i,k)} - \mathcal{R}_{i,k}| \leq C_*\left(  \Delta t\right) ^{3/4}.
	\end{equation}
	\item[(ii)] If \( \theta_* \sqrt{\Delta t} \leq \mathcal{R}_{i,k} < \theta^* \sqrt{\Delta t} \) then \( k_{d}(i, k) =k\) and \( k_{u}(i, k) = k + 1\). Moreover, one has
	\begin{equation}\label{lemma2_28a}
	\mathcal{R}_{i+1,k_{d}(i,k)} - \mathcal{R}_{i,k} = -\sigma \sqrt{\mathcal{R}_{i,k} \Delta t} + \frac{\sigma^2}{4} \Delta t,
	\end{equation}
	and
	\begin{equation}\label{lemma2_28b}
	\mathcal{R}_{i+1,k_{u}(i,k)} - \mathcal{R}_{i,k} = \sigma \sqrt{\mathcal{R}_{i,k} \Delta t} + \frac{\sigma^2}{4} \Delta t.
	\end{equation}
\end{enumerate}			
			\end{lemma}}
		
		\myrb{
		\begin{proof}
			The proof of (i) is very similar to the original one. We star by pointing out that in this particular case, it has been proven that $k_d^{\mathrm{ACZ}}(i,k)=k$, so, with respect to our tree,  $k_d(i,k)=\max\left\lbrace k, k_{\mathrm{min}}(i+1) \right\rbrace$.
			\begin{itemize}
				\item If $\mathcal{R}_{i,k}>0$, then $\mathcal{R}_{i+1,k+1}>\mathcal{R}_{i,k}>0$ and $ k_{\mathrm{min}}(i+1)\leq k$. Then $k_d(i,k)=\max\left\lbrace k, k_{\mathrm{min}}(i+1,k) \right\rbrace =k_d^{\mathrm{ACZ}}(i,k)$  and the original proof about is not altered.
				\item If $\mathcal{R}_{i,k}=0$, then $k= k_{\mathrm{min}}(i,k)$ and $R_{i+1,k}\leq \mathcal{R}_{i,k}=0$, so $ k_{\mathrm{min}}(i+1,k)\geq k$. Moreover, since $k\geq k_{\mathrm{min}}(i,k)$, then $\mathcal{R}_{i,k+1}>0$ and $ k_{\mathrm{min}}(i+1,k)\leq k+1$, so $ k_{\mathrm{min}}(i+1,k)=k$ or $ k_{\mathrm{min}}(i+1,k)=k+1$.
				\begin{itemize}
					\item If $ k_{\mathrm{min}}(i+1,k)=k$, then $k_d(i,k)=\max\left\lbrace k,k \right\rbrace =k_d^{\mathrm{ACZ}}(i,k)$, and one continues as in the original proof.  
					\item If $ k_{\mathrm{min}}(i+1,k)=k+1$, then $k_d(i,k)=\max\left\lbrace k,k+1 \right\rbrace =k+1$ and $\mathcal{R}_{i+1,k_d^{\mathrm{ACZ}}(i,k)}=0$, so $|\mathcal{R}_{i+1,k_{d}(i,k)} - \mathcal{R}_{i,k}| =0\leq C_* \left(\Delta t \right) ^{3/4}$ and nothing changes for $k_u(i,k)$.
					\end{itemize}
			\end{itemize}
			The proof of (ii) is direct. In fact, in this case $\mathcal{R}_{i,k}>0$, so, as discussed before, $k_d(i,k)=k_d^{\mathrm{ACZ}}(i,k)$, and, by definition, $k_u(i,k)=k_d^{\mathrm{ACZ}}(i,k)+1=k+1$. The proof of equations \eqref{lemma2_28a} and \eqref{lemma2_28b} follows as in the original reasoning.
		\end{proof}		 
			\setcounter{lemma}{3}
	\begin{lemma}
		Let $r_*$ be fixed. Then there exists $h_2 = h_2(r_*, \sigma_{F}) < 1$ such that for every $\Delta t  < h_2$ and $(i, k)$ such that $\mathcal{R}_{i,k} \in [0, r_*]$ one has
		\begin{equation}
			j_d(i, j, k) = j-1 \quad \text{and} \quad j_u(i, j, k) = j + 1
		\end{equation}
	for all $j=2,\dots,j_{\max}-1$.
		As a consequence, for $\Delta t  < h_2$ and for every $(i, k)$ such that $r_{i,k} \in [0, r_*]$ one has
		\begin{equation}
			\mathcal{A}_{j_u(i,j,k)} - \mathcal{A}_{j} = \mathcal{A}_{j}\left(e^{\sigma_{F}\sqrt{\Delta t }} - 1\right) \quad \text{and} \quad \mathcal{A}_{j_d(i,j,k)} - \mathcal{A}_{j} = \mathcal{A}_{j}\left(e^{-\sigma_{F}\sqrt{\Delta t }} - 1\right).
		\end{equation}
		\end{lemma} 
	\begin{proof}
		First of all, we observe that $ \mathcal{A}_{j-1}=\mathcal{A}_{j}e^{-\sigma_{F}\sqrt{\Delta t}}$  and $ \mathcal{A}_{j+1}=\mathcal{A}_{j}e^{\sigma_{F}\sqrt{\Delta t}}$.
			Moreover, for $j^*<j$ one has $\mathcal{A}_{j^*}\leq\mathcal{A}_{j}<\mathcal{A}_{j}\left(1+\mathcal{R}_{i,k}\Delta t \right) $, so $j_d(i, j, k) =j-1$. Concerning the up movement, one has to prove that for $\Delta t$ sufficiently small one has  
		$$\mathcal{A}_{j}\left(1+\mathcal{R}_{i,k}\Delta t \right) \geq \mathcal{A}_{j}e^{\sigma_{F}\sqrt{\Delta t}}$$ or equivalently 	$$  \mathcal{R}_{i,k}\Delta t \geq  e^{\sigma_{F}\sqrt{\Delta t}}-1 .$$
		This is true since $$\left( e^{\sigma_{F}\sqrt{\Delta t}}-1\right)-\mathcal{R}_{i,k}\Delta t  \geq \sigma_{F}\sqrt{\Delta t}-\mathcal{R}_{i,k}\Delta t =\sqrt{\Delta t}\left( \sigma_{F}-\mathcal{R}_{i,k}\sqrt{\Delta t}\right). $$
		Finally, the last term of the previous equality is positive for $\Delta t<\left(\sigma_{F}/\mathcal{R}_{i,k} \right) ^2$.
	\end{proof}
\setcounter{proposition}{5}
\begin{proposition}
	\begin{enumerate}
		Let $r_* > 0$ and $A_* > 0$ be fixed and set $\mathcal{I}_* = \{(i,j,k) : 1<j<j_{\max},\mathcal{R}_{i,k} \leq r_*, \mathcal{A}_{j} \leq A_*\}$. Let $\theta_*$ be as in Lemma 2 and $(i,j,k) \in \mathcal{I}_*$. We set:
		\item[(i)] if $(i,j,k) \in \mathcal{I}_*$ and $r_{i,k} < \theta_* \sqrt{\Delta t}$ then
		\begin{align*}
			p_{u,u} &= p^{A}_{i,j,k}p^{R}_{i,k}, \\
			p_{u,d} &= p^{A}_{i,j,k}\left( 1-p^{R}_{i,k}\right) , \\
			q_{d,u} &=p^{R}_{i,k}\left( 1-p^{A}_{i,j,k}\right), \\
			q_{d,d} &= \left( 1-p^{A}_{i,j,k}\right) \left( 1-p^{R}_{i,k}\right);
		\end{align*}
		\item[(ii)] if $(i,j,k) \in \mathcal{I}_*$ and $r_{i,k} \geq \theta_* \sqrt{\Delta t}$ then the four transition probabilities are set as the solutions of linear system  \eqref{sistema}.
	\end{enumerate}
	Then there exists $h_3 < 1$ and a positive constant $C$ such that for every $\Delta t < h_3$ and $(i,j,k) \in \mathcal{I}_*$ the above probabilities are actually well defined.
\end{proposition}
The proof of this result is based on the properties demonstrated in the previous Lemmas 2 and 4. Below there are some details.
Case (i) is handled by an approximation: when $\mathcal{R}_{i,k}$ is close to zero, no correlation is assumed, which is exactly the case for $\mathcal{R}_{i,k}=0$. Case (ii) is solved by studying the linear system \eqref{sistema} directly: explicit formulae for the solutions can be determined and it is shown that for sufficiently small $\Delta t$, the four probabilities obtained are all non-negative.}

\medskip
\myrb{To conclude, we observe that, thanks to the results of Lemmas 2 and 4, and Proposition 6, \cite{Appolloni2015} prove the weak convergence of the discrete process $\bar{A},\bar{r}$ to the corresponding continuous processes.}


\end{document}